%% file: Main.tex
\documentclass[12pt,onecolumn,journal]{IEEEtran}

\usepackage{epstopdf}
\usepackage{graphicx}
\usepackage{stmaryrd}
\usepackage{amsmath,amsthm,amssymb}
\usepackage{multirow,url}
\usepackage{color,cite}
\usepackage{algorithm}
\usepackage{algorithmicx}
\usepackage{algpseudocode}
\usepackage{setspace}
\usepackage{rotating}
\newtheorem{lem}{Lemma}

\newtheorem{cor}{Corollary}

\newtheorem{thm}{Theorem}
\newtheorem{defn}{Definition}

\def\sgn{\mathop{\rm sgn}}

\ifodd 1
\newcommand{\com}[1]{\textbf{\color{red} (COMMENT: #1)}} 
\newcommand{\comg}[1]{\textbf{\color{green} (COMMENT: #1)}}
\newcommand{\response}[1]{\textbf{\color{magenta} (RESPONSE: #1)}} 
\else

\newcommand{\com}[1]{}
\newcommand{\comg}[1]{}
\newcommand{\response}[1]{}
\fi

\begin{document}

\title{Database-assisted Distributed Spectrum Sharing}

\author{Xu Chen and Jianwei Huang \\Department of Information Engineering, The Chinese University of Hong Kong \\Email:\{cx008,jwhuang\}@ie.cuhk.edu.hk \thanks{*Part of the results has been published at IEEE ICDCS 2012 \cite{ICDCS2012}.}\vspace{-0.8cm}}

\maketitle
\pagestyle{empty}
\thispagestyle{empty}
\allowdisplaybreaks

\input{Abstract}
\input{Introduction}

\input{AP}

\input{AP2}
\input{User}

\input{Sim}

\input{Conclusion}

\bibliographystyle{ieeetran}
\bibliography{Database}

\end{document}

%% file: Abstract.tex
\begin{abstract}
According to FCC's ruling for white-space spectrum access, white-space
devices are required to query a database to determine the spectrum
availability. In this paper, we study the database-assisted distributed white-space access point (AP) network design.
We first model the cooperative and non-cooperative channel selection problems among the APs as the system-wide throughput optimization and non-cooperative
AP channel selection games, respectively, and design distributed AP channel selection
algorithms that achieve system optimal point and Nash equilibrium, respectively. We then propose a state-based
game formulation for the distributed AP association problem
of the secondary users by taking the cost of mobility into account.
We show that the state-based distributed AP association game has the
finite improvement property, and design a distributed AP association
algorithm that can converge to a state-based Nash equilibrium. Numerical
results show that the algorithm is robust to the perturbation by secondary
users' dynamical leaving and entering the system.
\end{abstract} 

%% file: Introduction.tex
\section{Introduction}

The most recent FCC ruling requires that TV white-space devices must
rely on a geo-location database to determine the spectrum availability
\cite{FCC}. In such a database-assisted architecture, the incumbents (primary licensed holders of TV spectrum)
provide the database with the up-to-date information including TV tower transmission
parameters and TV receiver protection requirements. Based on this information, the database will be able to tell a white-space device (secondary users (SUs) of TV spectrum) vacant TV channels at a particular location, given the white-space device's transmission parameters such as the transmission power.

Although the database-assisted approach obviates the need of spectrum
sensing, the task of developing a comprehensive and reliable database-assisted white-space network
system remains challenging \cite{db4}. Motivated by the
successful deployments of Wi-Fi over the unlicensed ISM bands, in
this paper we consider an infrastructure-based white-space network (see Figure
\ref{fig:Distributed-Spectrum-Sharing} for an illustration),
where there are multiple secondary access points (APs) operating on white spaces. Such an infrastructure-based architecture has been adopted in IEEE 802.22 standard \cite{IEEE} and Microsoft Redmond campus white-space networking experiment \cite{db4}. More specifically, each
AP first sends the required information such as its location and the transmission power to the database via wire-line connections. The database
then feeds back the set of vacant TV channels at the location of each AP. Afterwards, an AP chooses one feasible channel to serve the secondary users (i.e., unlicensed white-space user devices) within its transmission range.

The key challenges for such an infrastructure-based white-space network
design are twofold (see Figure \ref{fig:System-Architecture} for
an illustration). First, in the AP tier, each AP must choose a proper vacant channel
to operate in order to avoid severe interference with other APs. Second, in the SU tier, when an AP is overloaded, a secondary
user can improve its throughput by moving to and associating with
another AP with less contending users. Each secondary user hence needs
to decide which AP to associate with.

\begin{figure}[tt]
\centering
\includegraphics[scale=0.55]{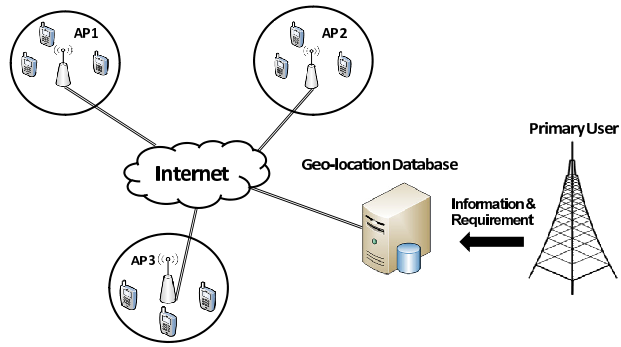}
\caption{\label{fig:Distributed-Spectrum-Sharing}Distributed spectrum sharing
with geo-location database}
\end{figure}

\begin{figure}[tt]
\centering
\includegraphics[scale=0.65]{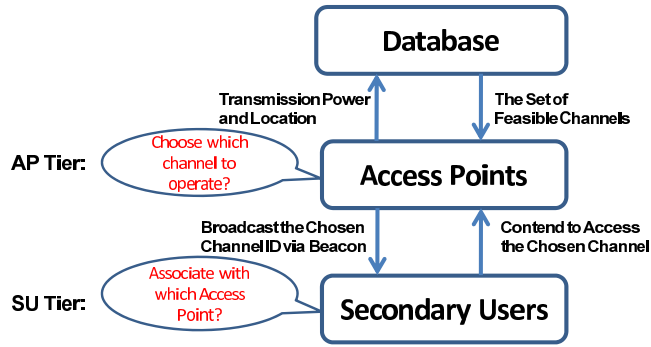}
\caption{\label{fig:System-Architecture}System architecture}
\end{figure}

\begin{table}
\caption{\label{tab:Summary-of-the}Summary of the results}
\small
\centering
\begin{tabular}{|c|c|c|c|}
\hline
 & Problem & Cooperative AP Channel Selection & Non-Cooperative AP Channel Selection\tabularnewline
\cline{2-4}
AP Tier & Formulation & System-wide Throughput Optimization & Non-Cooperative AP Channel Selection Game\tabularnewline
\cline{2-4}
 & Algorithm & Cooperative AP Channel Selection Algorithm & Non-Cooperative AP Channel Selection Algorithm\tabularnewline
\hline
 & Problem & \multicolumn{2}{c|}{Distributed AP Association by SUs}\tabularnewline
\cline{2-4}
SU Tier & Formulation & \multicolumn{2}{c|}{Distributed AP Association Game}\tabularnewline
\cline{2-4}
 & Algorithm & \multicolumn{2}{c|}{Distributed AP Association Algorithm}\tabularnewline
\hline
\end{tabular}
\end{table}

In this paper, for the AP tier, we first consider the scenario that all the APs are owned by one network operator and hence the APs are cooperative. We formulate the cooperative AP channel selection problem as the system-wide throughput optimization problem. We then consider the scenario that the APs are owned by different network operators and the interest of APs is not aligned. We model the distributed channel selection problem among the APs as a non-cooperative AP channel selection game. For the SU tier, we propose a state-based game framework to model the distributed AP association
problem of the secondary users by taking the cost of mobility into
account. The main results and contributions of this paper are as follows (please refer to Table \ref{tab:Summary-of-the} for a summary):
\begin{itemize}
\item \emph{General formulation}: We formulate the cooperative and non-cooperative channel selection
problems among the APs as system-wide throughput optimization and non-cooperative AP channel selection game, respectively, based
on the physical interference model \cite{Gupta00thecapacity}. We then propose a
state-based game framework to formulate the distributed AP association
problem of the secondary users and explicitly take the cost of
mobility into account.
\item \emph{Existence of equilibrium solution and finite improvement property}: For the cooperative AP channel selection problem, the interest of APs is aligned and the system optimal solution that maximizes system-wide throughput always exists. For the non-cooperative AP channel selection game, we show that it is
a potential game, and hence it has a Nash equilibrium and the finite
improvement property. For the state-based distributed AP association
game, we show that it also has a state-based Nash equilibrium and
the finite improvement property.
\item \emph{Distributed algorithms for achieving equilibrium}: For the cooperative AP channel selection problem, we propose a cooperative channel selection algorithm that maximizes the system-wide throughput. For the non-cooperative AP channel selection game, we propose a non-cooperative
AP channel selection algorithm that achieves a Nash equilibrium of
the game. For the state-based distributed AP association game, we
design a distributed AP association algorithm that converges to a
state-based Nash equilibrium. Numerical results show that the algorithm
is robust to the perturbation by secondary users' dynamical leaving
and entering the system.
\end{itemize}

The rest of the paper is organized as follows. We introduce the cooperative and non-cooperative
AP channel selection problems, and propose the cooperative and non-cooperative AP channel
selection algorithms in Sections \ref{sec:Distributed-AP-Channel1}
and \ref{sec:Distributed-AP-Channel}, respectively. We present
the distributed AP association game and distributed AP association algorithm in Section \ref{sec:Distributed-AP-Association-1}.
We illustrate the performance
of the proposed mechanisms through numerical results in Section \ref{sec:Simulation-Results},
and finally introduce the related work and conclude in Sections \ref{sec:Related-Work}
and \ref{sec:Conclusion}, respectively.

%% file: AP.tex
\section{Cooperative AP Channel Selection}\label{sec:Distributed-AP-Channel1}

\subsection{System Model}

We first introduce the system model for the cooperative channel selection problem among
the APs in the AP tier. Let $\mathcal{M}=\{1,2,...,M\}$ denote the set of TV
channels, and $B$ denote the bandwidth of each channel (e.g., $B=6$
MHz in the United States and $B=8$ MHz in the European Union). We consider a set $\mathcal{N}=\{1,2,...,N\}$
of APs that operate on the white spaces. Each AP $n\in\mathcal{N}$
has a specified transmission power $P_{n}$ based on its coverage and primary user protection requirements.

Each AP $n$ can acquire the information of the vacant channels
at its location from the geo-location database. We denote $\mathcal{M}_{n}\subseteq\mathcal{M}$
as the set of feasible channels of AP $n$, $a_{n}\in\mathcal{M}_{n}$
as the channel chosen by AP $n$\footnote{Following the conventions in IEEE 802.22 standard \cite{IEEE} and Microsoft Redmond campus white-space networking experiment \cite{db4}, we consider the case that each AP can select one channel to operate on. The case that each AP can select multiple channels to operate on will be considered in a future work.}, and $\boldsymbol{a}=(a_{1},...,a_{N})$ as the
channel selection profile of all APs. Then the worse-case down-link
throughput (i.e., the throughput at the boundary of the coverage area) of AP $n$ can be computed according to the physical interference
model \cite{Gupta00thecapacity} as\begin{equation}
U_{n}(\boldsymbol{a})=B\log_{2}\left(1+\frac{P_{n}/d_{n}^{\theta}}{\omega_{a_{n}}^{n}+\sum_{i\in\mathcal{N}/\{n\}:a_{i}=a_{n}}P_{i}/d_{in}^{\theta}}\right),\label{eq:U}\end{equation}
where $\theta$ is the path loss factor, $d_{n}$ denotes the radius of the coverage area of AP $n$, and $d_{in}$ denotes
the distance between AP $i$ and the benchmark location at the boundary of the coverage area of AP $n$. Furthermore, $\omega_{a_{n}}^{n}$
denotes the background noise power including the interference from incumbent users
on the channel $a_{n}$, and $\sum_{i\in\mathcal{N}/\{n\}:a_{i}=a_{n}}P_{i}/d_{in}^{\theta}$
denotes the accumulated interference from other APs that choose the
same channel $a_{n}$. Note that we assume that all APs only try to maximize the worse-case throughputs by proper channel selections, which do not depend on the number of its associated users. However, the secondary users can increase their data rates by moving to and associating with a less congested AP (see Section \ref{sec:Distributed-AP-Association-1} for detailed discussions). Note that our model also applies to the up-link case if the secondary users within an AP transmit with roughly the same power level.

\subsection{Cooperative AP Channel Selection Algorithm}

We first consider the case that all the APs try to maximize the system-wide throughput cooperatively. Such a cooperation is feasible when all the APs are owned by the same network
operator. For example, the APs that are deployed in a university campus can coordinate to maximize the entire campus network throughput. Formally, the APs need to collectively determine the optimal channel selection profile $\boldsymbol{a}$ such that the system-wide throughput is maximized, i.e., \begin{align}\max_{\boldsymbol{a}\in\Theta\triangleq\Pi_{n=1}^{N}\mathcal{M}_{n}}\sum_{n=1}^{N}U_{n}(\boldsymbol{a}).\label{cooperS} \end{align}
The problem (\ref{cooperS}) is a combinatorial optimization problem of finding the optimal channel selection profile over the discrete solution space $\Theta$. In general, such a problem is very challenging to solve exactly especially when the size of network is large (i.e., the solution space $\Theta$ is large).

We next propose a cooperative channel selection algorithm that can approach
the optimal system-wide throughput approximatively. To proceed, we first write the problem (\ref{cooperS}) into the following equivalent problem:
\begin{equation}
\max_{(q_{\boldsymbol{a}}:\boldsymbol{a}\in\Theta)}\sum_{\boldsymbol{a}\in\Theta}q_{\boldsymbol{a}}\sum_{n=1}^{N}U_{n}(\boldsymbol{a}),\label{eq:P12}\end{equation}
where $q_{\boldsymbol{a}}$ is the probability that channel selection profile $\boldsymbol{a}$ is adopted. Obviously, the optimal solution to problem (\ref{eq:P12}) is to choose the optimal channel selection profile with probability one. It is known from \cite{chen2010markov22} that problem (\ref{eq:P12}) can be approximated by the following convex optimization problem:  \begin{equation}
\max_{(q_{\boldsymbol{a}}:\boldsymbol{a}\in\Theta)}\sum_{\boldsymbol{a}\in\Theta}q_{\boldsymbol{a}}\sum_{n=1}^{N}U_{n}(\boldsymbol{a})-\frac{1}{\gamma}\sum_{\boldsymbol{a}\in\Theta}q_{\boldsymbol{a}}\log q_{\boldsymbol{a}},\label{eq:P1}\end{equation}
where $\gamma$ is the parameter that controls the approximation ratio. We see that when $\gamma\rightarrow\infty$, the problem
(\ref{eq:P1}) becomes exactly the same as problem (\ref{eq:P12}).
That is, when $\gamma\rightarrow\infty,$ the optimal point
$\boldsymbol{a}^{*}$ that maximizes the system throughput $\sum_{n=1}^{N}U_{n}(\boldsymbol{a})$
will be selected with probability one. A nice property of such an approximation in (\ref{eq:P1}) is that we can obtain the close-form solution, which enables the distributed algorithm design later. More specifically, by the KKT condition \cite{CO}, we can derive the optimal solution to problem (\ref{eq:P1}) as\begin{equation}
q_{\boldsymbol{a}}^{*}=\frac{\exp\left(\gamma\sum_{n=1}^{N}U_{n}(\boldsymbol{a})\right)}{\sum_{\boldsymbol{a}^{'}\in\Theta}\exp\left(\gamma\sum_{n=1}^{N}U_{n}(\boldsymbol{a}^{'})\right)}.\label{eq:P3}\end{equation}

\begin{figure}[tt]
\centering
\includegraphics[scale=0.8]{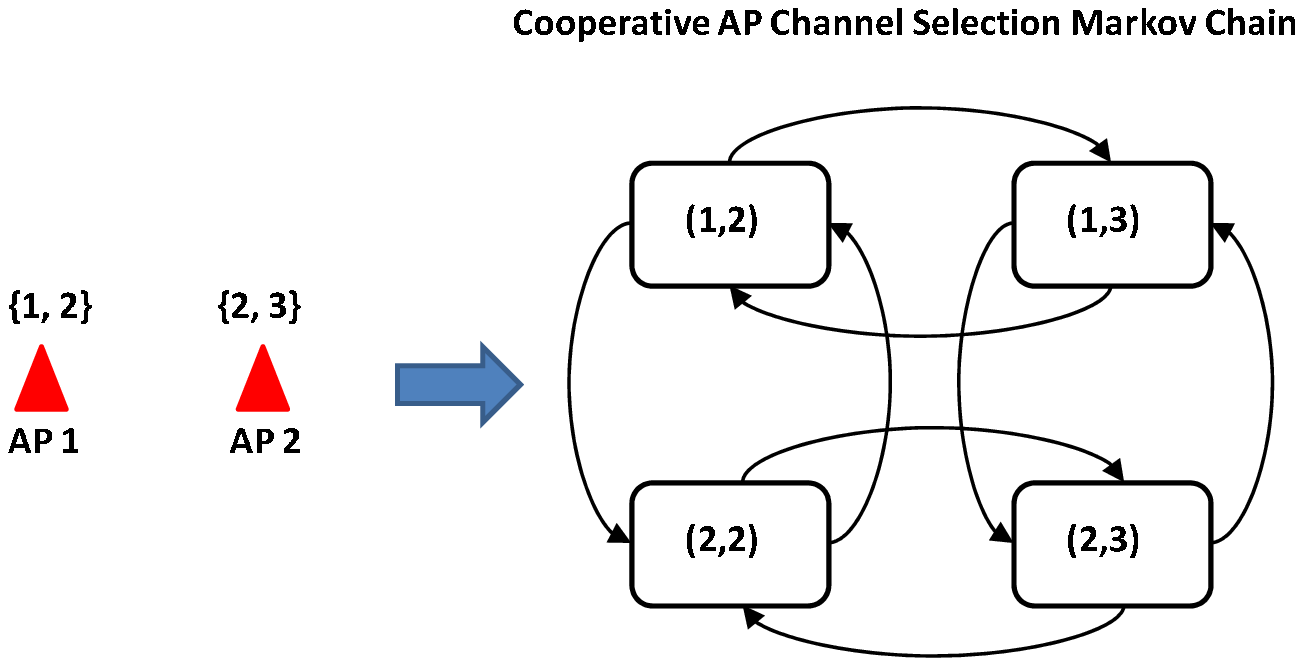}
\caption{\label{fig:APMarkov}System state transition diagram of the cooperative AP channel selection Markov chain by two APs. Figure on the left hand-side details the vacant channels of two APs. For example, AP $1$ can choose channels $1$ and $2$ to transmit. Figure on the right hand-side shows the transition diagram of the Markov chain, and $(a_{1},a_{2})$ denotes the system state with $a_{1}$ and $a_{2}$ being the channels chosen by APs $1$ and $2$, respectively. The direct transition between two system states is feasible if they are connected by a link.}
\end{figure}

\begin{algorithm}[tt]
\begin{algorithmic}[1]
\State \textbf{initialization:}
\State \hspace{0.4cm} \textbf{choose} an initial channel $a_{n}\in\mathcal{M}_{n}$ randomly for each AP $n\in\mathcal{N}$.
\State \hspace{0.4cm} \textbf{acquire} the information of initial channel selections, transmission powers, and geo-locations from other APs by each AP $n\in\mathcal{N}$.
\State \textbf{end initialization\newline}

\Loop{ for each iteration:}
        \State Database \textbf{selects} an AP randomly and \textbf{informs} the selected AP to update its channel selection.
        \For{each AP $n\in\mathcal{N}$ in parallel}
            \If{ the update command is received from the database}
                \State \textbf{calculate} the system throughput $\sum_{n=1}^{N}U_{n}(a_{n},a_{-n})$ for each feasible channel selection $a_{n}\in\mathcal{M}_{n}$.
                \State \textbf{select} a channel $a_{n}\in\mathcal{M}_{n}$ with a probability of $\frac{\exp\left(\gamma \sum_{n=1}^{N}U_{n}(a_{n},a_{-n})\right)}{\sum_{a^{'}\in\mathcal{M}_{n}}\exp\left(\gamma \sum_{n=1}^{N}U_{n}(a^{'},a_{-n})\right)}.$
                \State \textbf{broadcast} the chosen channel $a_{n}$ to other APs.
            \Else{ \textbf{select} the original channel.}
            \EndIf
        \EndFor
\EndLoop
\end{algorithmic}
\caption{\label{alg:Online-Distributed-Channel-2} Cooperative AP Channel Selection
Algorithm}
\end{algorithm}

Similarly to the spatial adaptive play in \cite{young2001individual} and Gibbs sampling in \cite{kauffmann2007measurement}, we then design a cooperative AP channel selection algorithm by
carefully coordinating APs' asynchronous channel selection
updates to form a discrete-time Markov chain (with the system state as the channel selection profile $\boldsymbol{a}$ of all APs). As long as the Markov chain converges to the stationary distribution as given in (\ref{eq:P3}), we can approach the optimal channel selection profile that maximizes the system-wide throughput by setting a large enough parameter $\gamma$. The details of the
algorithm are given in Algorithm \ref{alg:Online-Distributed-Channel-2}.
Here APs' asynchronous channel selection updates are scheduled by
the database. In each iteration, one AP will be randomly chosen to
update its channel selection. In this case, the direct transitions between two system states $\boldsymbol{a}$ and $\boldsymbol{a}^{'}$  are feasible if these two system states differ by one and only one AP channel selection. As an example, the system state transition diagram of the cooperative AP channel selection Markov chain by two APs is shown in Figure \ref{fig:APMarkov}. We also denote the set of system states that can be transited directly from the state $\boldsymbol{a}$ as $\Lambda_{\boldsymbol{a}}\triangleq\{\boldsymbol{a}^{'}\in\Theta:|\{\boldsymbol{a}\cup\boldsymbol{a}^{'}\} / \{\boldsymbol{a}\cap\boldsymbol{a}^{'}\}|=2 \}$, where $|\cdot|$ denotes the size of a set.

Since each AP will be selected to update with a probability of $\frac{1}{N}$ and the selected AP will randomly choose a channel with
a probability proportional to $\exp\left(\gamma\sum_{n=1}^{N}U_{n}(\boldsymbol{a})\right)$, then if $\boldsymbol{a}^{'}\in\Lambda_{\boldsymbol{a}}$, the probability that the Markov chain transits from state $\boldsymbol{a}$
to $\boldsymbol{a}^{'}$ is given as \begin{align}
q_{\boldsymbol{a},\boldsymbol{a}^{'}}=\frac{1}{N}\frac{\exp\left(\gamma\sum_{n=1}^{N}U_{n}(a_{n}^{'},a_{-n})\right)}{\sum_{a^{'}\in\mathcal{M}_{n}}\exp\left(\gamma\sum_{n=1}^{N}U_{n}(a^{'},a_{-n})\right)}.\label{pp22}\end{align}
Otherwise, we have $q_{\boldsymbol{a},\boldsymbol{a}^{'}}=0$. We show in Theorem \ref{cooperAP} that the cooperative AP channel selection Markov chain
is time reversible. Time reversibility means that when tracing the Markov chain backwards, the stochastic
behavior of the reverse Markov chain remains the same. A nice property of a time reversible Markov
chain is that it always admits a unique stationary distribution, which guarantees the convergence of the
cooperative AP channel selection algorithm.

\begin{thm} \label{cooperAP}
The cooperative AP channel selection algorithm induces a time-reversible Markov chain with the unique stationary distribution as given in (\ref{eq:P3}).
\end{thm}
\begin{proof}
As mentioned, the system state of the cooperative AP channel selection
Markov chain is defined as the channel selection profile $\boldsymbol{a}\in\Theta$
of all APs. Since it is possible to get from any state to any other
state within finite steps of transition,
the AP channel selection Markov chain is hence irreducible and has a stationary distribution.

We then show that the Markov chain is time reversible by showing that the distribution in (\ref{eq:P3}) satisfies the following
detailed balance equations:\begin{equation}
q_{\boldsymbol{a}}^{*}q_{\boldsymbol{a},\boldsymbol{a}^{'}}=q_{\boldsymbol{a}^{'}}^{*}q_{\boldsymbol{a}^{'},\boldsymbol{a}}, \forall \boldsymbol{a}, \boldsymbol{a}^{'}\in\Theta.\label{eq:dddd}\end{equation}
To see this, we consider the following two cases:

1) If $\boldsymbol{a}^{'}\notin\Lambda_{\boldsymbol{a}}$, we have $q_{\boldsymbol{a},\boldsymbol{a}^{'}}=q_{\boldsymbol{a},\boldsymbol{a}^{'}}=0$ and the equation (\ref{eq:dddd}) holds.

2) If $\boldsymbol{a}^{'}\in\Lambda_{\boldsymbol{a}}$, according to (\ref{eq:P3}) and (\ref{pp22}), we have\begin{align*}
q_{\boldsymbol{a}}^{*}q_{\boldsymbol{a},\boldsymbol{a}^{'}} & = \frac{\exp\left(\gamma\sum_{n=1}^{N}U_{n}(\boldsymbol{a})\right)}{\sum_{\tilde{\boldsymbol{a}}\in\Theta}\exp\left(\gamma\sum_{n=1}^{N}U_{n}(\tilde{\boldsymbol{a}})\right)}
\frac{1}{N}\frac{\exp\left(\gamma\sum_{n=1}^{N}U_{n}(\boldsymbol{a}^{'})\right)}{\sum_{a^{'}\in\mathcal{N}_{n}}\exp\left(\gamma\sum_{n=1}^{N}U_{n}(a^{'},a_{-n})\right)}\\
 & = \frac{\exp\left(\gamma\sum_{n=1}^{N}U_{n}(\boldsymbol{a}^{'})\right)}{\sum_{\tilde{\boldsymbol{a}}\in\Theta}\exp\left(\gamma\sum_{n=1}^{N}U_{n}(\tilde{\boldsymbol{a}})\right)}
\frac{1}{N}\frac{\exp\left(\gamma\sum_{n=1}^{N}U_{n}(\boldsymbol{a})\right)}{\sum_{a^{'}\in\mathcal{N}_{n}}\exp\left(\gamma\sum_{n=1}^{N}U_{n}(a^{'},a_{-n})\right)}\\
 & = q_{\boldsymbol{a}^{'}}^{*}q_{\boldsymbol{a}^{'},\boldsymbol{a}}.\end{align*}
The cooperative AP channel selection Markov chain is hence time-reversible
and has the unique stationary distribution as given in (\ref{eq:P3}).\end{proof}

According to Theorem \ref{cooperAP}, we can approach the system optimal point that maximizes the system-wide throughput by setting $\gamma\rightarrow\infty$ in the cooperative AP channel selection algorithm. However, in practice we can only implement a finite value
of $\gamma$  such that $\exp(\gamma\sum_{n=1}^{N}U_n(\boldsymbol{a}))$ does not exceed the range of the largest predefined real number on a computer. Let $\bar{S}=\sum_{a\in\Theta}q_{a}^{*}\sum_{n=1}^{N}U_n(\boldsymbol{a})$ be
the expected potential by Algorithm \ref{alg:Online-Distributed-Channel-2} and $S^{*}=\max_{a\in\Theta}\sum_{n=1}^{N}U_n(\boldsymbol{a})$
be the global optimal potential. We show in Theorem \ref{thm:For-the-distributed-2}
that, when a large eough $\gamma$ is adopted, the performance
gap between $\bar{S}$ and $S^{*}$ is very small.
\begin{thm}
\label{thm:For-the-distributed-2}For the cooperative AP channel selection
algorithm, we have that \[
0\leq S^{*}-\bar{S}\leq\frac{1}{\gamma}\ln|\Theta|,\]
where $|\Theta|$ denotes the number of feasilbe channel selection
profiles of all APs.\end{thm}
\begin{proof}
First of all, we must have that $S^{*}\geq\bar{S}$. According to (\ref{cooperS}), (\ref{eq:P1}), and (\ref{eq:P3}), we then have that\begin{equation}
\max_{(q_{\boldsymbol{a}}:\boldsymbol{a}\in\Theta)}\sum_{a\in\Theta}q_{a}\sum_{n=1}^{n}U_{n}(\boldsymbol{a})\leq\max_{(q_{\boldsymbol{a}}:\boldsymbol{a}\in\Theta)}\sum_{a\in\Theta}q_{a}\sum_{n=1}^{n}U_{n}(\boldsymbol{a})-\frac{1}{\gamma}\sum_{a\in\Theta}q_{a}\ln q_{a},\label{eq:hhhhh1}\end{equation}
which is due to the fact that $0\leq-\frac{1}{\gamma}\sum_{a\in\Theta}q_{a}\ln q_{a}\leq\frac{1}{\gamma}\ln|\Theta|$.
Since $q_{a}^{*}$ is the optimal solution to (\ref{eq:P1}) and $S^{*}=\max_{(q_{\boldsymbol{a}}:\boldsymbol{a}\in\Theta)}\sum_{a\in\Theta}q_{a}\sum_{n=1}^{n}U_{n}(\boldsymbol{a})$,
according to (\ref{eq:hhhhh1}), we know that\begin{eqnarray*}
S^{*} & \leq & \sum_{a\in\Theta}q_{a}^{*}\sum_{n=1}^{n}U_{n}(\boldsymbol{a})-\frac{1}{\gamma}\sum_{a\in\Theta}q_{a}^{*}\ln q_{a}^{*}\\
 & \leq & \sum_{a\in\Theta}q_{a}^{*}\sum_{n=1}^{n}U_{n}(\boldsymbol{a})+\frac{1}{\gamma}\ln|\Theta|\\
 & \leq & \bar{S}+\frac{1}{\gamma}\ln|\Theta|,\end{eqnarray*}
which completes the proof.
\end{proof}

We then analyze the computational complexity of the algorithm. In each iteration, one AP will be chosen for the channel selection update. Line $9$ involves the summation of the throughputs of $N$ APs for $\mathcal{M}_{n}$ channels. Since $|\mathcal{M}_n|\leq M$, this step has the complexity of $\mathcal{O}(NM)$. Line $10$ involves at most $M$ summation and division operations and hence has a complexity of $\mathcal{O}(M)$. Line $11$ has a complexity of $\mathcal{M}(1)$. Suppose that it takes $C$ iterations for the algorithm to converge. Then total computational complexity of the algorithm is $\mathcal{O}(CNM)$. Similarly, the space complexity of the algorithm is $\mathcal{O}(N^2+NM)$.

%% file: AP2.tex
\section{\label{sec:Distributed-AP-Channel}Non-cooperative AP Channel Selection}

We next consider the case that the APs are owned by different network operators. Unlike the previous case where the interest of the APs is aligned in the cooperative channel selection, here each AP is generally selfish and only concerns about its own throughput maximization. Formally, given other APs'
channel selections $a_{-n}$, the problem faced by an AP $n$ is to choose a proper channel to maximize its own throughput, i.e.,\[
\max_{a_{n}\in\mathcal{M}_{n}}U_{n}(a_{n},a_{-n}),\forall n\in\mathcal{N}.\]
The non-cooperative nature of the channel selection problem naturally
leads to a formulation based on game theory, such that each AP
can self organize into a mutually acceptable channel selection (\textbf{Nash
equilibrium}) $\boldsymbol{a}^{*}=(a_{1}^{*},a_{2}^{*},...,a_{N}^{*})$ with\[
a_{n}^{*}=\arg\max_{a_{n}\in\mathcal{M}_{n}}U_{n}(a_{n},a_{-n}^{*}),\forall n\in\mathcal{N}.\]

\subsection{Non-Cooperative AP Channel Selection Game}

We now formulate the non-cooperative channel selection problem as a strategic
game $\\ \Gamma=(\mathcal{N},\{\mathcal{M}_{n}\}_{n\in\mathcal{N}},\{U_{n}\}_{n\in\mathcal{N}})$,
where $\mathcal{N}$ is the set of APs, $\mathcal{M}_{n}$ is the
set of strategies for AP $n$, and $U_{n}$ is the payoff function
of AP $n$. We refer this as the non-cooperative AP channel selection game
in the sequel.

We can show that it is a potential game, which is defined as
\begin{defn}[\textbf{Potential Game} \!\!\cite{Monderer1996}]
A game is called a potential game if it admits a potential function
$\Phi(\boldsymbol{a})$ such that for every $n\in\mathcal{N}$ and $a_{-n}\in\prod_{i\ne n}\mathcal{M}_{i}$,\begin{eqnarray*}
\sgn\left(\Phi(a_{n}^{'},a_{-n})-\Phi(a_{n},a_{-n})\right)=\sgn\left(U_{n}(a_{n}^{'},a_{-n})-U_{n}(a_{n},a_{-n})\right),\end{eqnarray*}
where $\sgn(\cdot)$ is the sign function defined as\[
\sgn(z)=\begin{cases}
1 & \mbox{if }z>0,\\
0 & \mbox{if }z=0,\\
-1 & \mbox{if }z<0.\end{cases}\]
\end{defn}

\begin{defn}[\textbf{Better Response Update} \!\!\cite{Monderer1996}]
The event where a player $n$ changes to an action $a_{n}^{'}$ from
the action $a_{n}$ is a better response update if and only if $U_{n}(a_{n}^{'},a_{-n})>U_{n}(a_{n},a_{-n})$.
\end{defn}
An appealing property of the potential game is that it admits the
finite improvement property, such that any asynchronous better response
update process (i.e., no more than one player updates the strategy at any given
time) must be finite and leads to a Nash equilibrium \cite{Monderer1996}.

To show that the non-cooperative AP channel selection game $\Gamma$ is
a potential game, we now consider a closely related game $\tilde{\Gamma}=(\mathcal{N},\{\mathcal{M}_{n}\}_{n\in\mathcal{N}},\{\tilde{U}_{n}\}_{n\in\mathcal{N}})$, where the new payoff functions are \begin{equation}
\tilde{U}_{n}(\boldsymbol{a})=\frac{P_{n}/d_{n}^{\theta}}{\omega_{a_{n}}^{n}+\sum_{i\in\mathcal{N}/\{n\}:a_{i}=a_{n}}P_{i}/d_{in}^{\theta}}.\label{eq:U2}\end{equation}
Obviously, the utility function $U_{n}(\boldsymbol{a})$ can be obtained from the
utility function $\tilde{U}_{n}(\boldsymbol{a})$ by the following monotone transformation\begin{equation}
U_{n}(\boldsymbol{a})=B\log_{2}\left(1+\tilde{U}_{n}(\boldsymbol{a})\right).\label{eq:U3}\end{equation}
Due to the property of monotone transformation, we have
\begin{lem}
\label{lem:If-the-modified1}If the modified game $\tilde{\Gamma}$
is a potential game, then the original non-cooperative AP channel selection
game $\Gamma$ is also a potential game with the same potential function.\end{lem}
\begin{proof}
Since $f(x)=B\log_{2}(1+x)$ is a monotonically strictly increasing
function, we have that\begin{eqnarray*}
   \sgn\left(U_{n}(a_{n}^{'},a_{-n})-U_{n}(a_{n},a_{-n})\right) = \sgn\left(\tilde{U}_{n}(a_{n}^{'},a_{-n})-\tilde{U}_{n}(a_{n},a_{-n})\right).\end{eqnarray*}
If the modified game $\tilde{\Gamma}$ is a potential game with a
potential function $\Phi$ such that\begin{eqnarray*}
\sgn\left(\Phi(a_{n}^{'},a_{-n})-\Phi(a_{n},a_{-n})\right) = \sgn\left(\tilde{U}_{n}(a_{n}^{'},a_{-n})-\tilde{U}_{n}(a_{n},a_{-n})\right),\end{eqnarray*}
then we must also have that\begin{eqnarray*}
 \sgn\left(\Phi(a_{n}^{'},a_{-n})-\Phi(a_{n},a_{-n})\right) = \sgn\left(U_{n}(a_{n}^{'},a_{-n})-U_{n}(a_{n},a_{-n})\right),\end{eqnarray*}
which completes the proof.
\end{proof}
For the modified game $\tilde{\Gamma}$, we show in Theorem \ref{thm:The-modified-distributed2} that
it is a potential game with the following potential function\begin{equation}
\Phi(\boldsymbol{a})=-\sum_{i}\sum_{j\ne i}\frac{P_{i}P_{j}}{d_{ij}^{\theta}}I_{\{a_{i}=a_{j}\}}-2\sum_{i=1}^{N}P_{i}\omega_{a_{i}}^{i},\label{eq:P11}\end{equation}
where $I_{\{a_{i}=a_{j}\}}=1$ if $a_{i}=a_{j}$, and $I_{\{a_{i}=a_{j}\}}=0$
otherwise.
\begin{thm}
\label{thm:The-modified-distributed2}The modified game $\tilde{\Gamma}$
is a potential game with the potential function $\Phi(\boldsymbol{a})$ as given in (\ref{eq:P11}).\end{thm}
\begin{proof}
Suppose that an AP $k$ changes its channel $a_{k}$ to $a_{k}^{'}$ such
that the strategy profile changes from $\boldsymbol{a}$ to $\boldsymbol{a}^{'}$. We have that\begin{align*}
 &  \Phi(\boldsymbol{a}^{'})-\Phi(\boldsymbol{a})\\
  = & -\sum_{j\ne k}\frac{P_{k}P_{j}}{d_{kj}^{\theta}}I_{\{a_{k}^{'}=a_{j}\}}+\sum_{j\ne k}\frac{P_{k}P_{j}}{d_{kj}^{\theta}}I_{\{a_{k}=a_{j}\}}
  -\sum_{i\ne k}\frac{P_{i}P_{k}}{d_{ik}^{\theta}}I_{\{a_{i}=a_{k}^{'}\}}+\sum_{i\ne k}\frac{P_{i}P_{k}}{d_{ik}^{\theta}}I_{\{a_{i}=a_{k}\}}
  -2P_{k}\omega_{a_{k}^{'}}^{k}+2P_{k}\omega_{a_{k}}^{k}.\end{align*}
Since $d_{ij}=d_{ji}$, we thus have that\begin{align*}
   & \Phi(\boldsymbol{a}^{'})-\Phi(\boldsymbol{a})\\
  = & -2\sum_{i\ne k}\frac{P_{i}P_{k}}{d_{ik}^{\theta}}I_{\{a_{i}=a_{k}^{'}\}}+2\sum_{i\ne k}\frac{P_{i}P_{k}}{d_{ik}^{\theta}}I_{\{a_{i}=a_{k}\}}
  -2P_{k}\omega_{a_{k}^{'}}^{k}+2P_{k}\omega_{a_{k}}^{k}\\
  = & -2P_{k}\left(\sum_{i\ne k:I_{\{a_{i}=a_{k}^{'}\}}}\frac{P_{i}}{d_{ik}^{\theta}}+\omega_{k,a_{k}^{'}}\right)
  +2P_{k}\left(\sum_{i\ne k:I_{\{a_{i}=a_{k}\}}}\frac{P_{i}}{d_{ik}^{\theta}}+\omega_{k,a_{k}}\right)\\
  = & 2d_{k}^{\theta}\left(\sum_{i\ne k:I_{\{a_{i}=a_{k}^{'}\}}}\frac{P_{i}}{d_{ik}^{\theta}}+\omega_{k,a_{k}^{'}}\right)\left(\sum_{i\ne k:I_{\{a_{i}=a_{k}\}}}\frac{P_{i}}{d_{ik}^{\theta}}+\omega_{k,a_{k}}\right)\\
   & \times\left(\frac{P_{k}/d_{k}^{\theta}}{\sum_{i\ne k:I_{\{a_{i}=a_{k}^{'}\}}}\frac{P_{i}}{d_{ik}^{\theta}}+\omega_{a_{k}^{'}}^{k}}-\frac{P_{k}/d_{k}^{\theta}}{\sum_{i\ne k:I_{\{a_{i}=a_{k}\}}}\frac{P_{i}}{d_{ik}^{\theta}}+\omega_{a_{k}}^{k}}\right)\\
  = & 2d_{k}^{\theta}\left(\sum_{i\ne k:I_{\{a_{i}=a_{k}^{'}\}}}\frac{P_{i}}{d_{ik}^{\theta}}+\omega_{a_{k}^{'}}^{k}\right)\left(\sum_{i\ne k:I_{\{a_{i}=a_{k}\}}}\frac{P_{i}}{d_{ik}^{\theta}}+\omega_{a_{k}}^{k}\right)\left(\tilde{U}_{k}(a_{k}^{'},a_{-k})-\tilde{U}_{k}(a_{k},a_{-k})\right),\end{align*}
which completes the proof.
\end{proof}
According to Lemma \ref{lem:If-the-modified1} and Theorem \ref{thm:The-modified-distributed2},
we know that
\begin{thm}
\label{thm:The-distributed-channel}The non-cooperative AP channel selection
game $\Gamma$ is a potential game, which has a Nash equilibrium and
the finite improvement property.
\end{thm}
The result in Theorem \ref{thm:The-distributed-channel} implies that
any asynchronous better response update is guaranteed to reach a
Nash equilibrium within a finite number of iterations. This motivates
the algorithm design in Section \ref{sec:Distributed-Channel-Selection}. Interestingly, according to the property of potential game, any channel selection
profile $\boldsymbol{a}$ that maximizes the potential function $\Phi(\boldsymbol{a})$
is a Nash equilibrium \cite{Monderer1996}. According to (\ref{eq:P11}), the profile $\boldsymbol{a}^{*}$ is also an efficient system-wide solution, since maximizing the potential function $\Phi(\boldsymbol{a})$ is equivalent to minimizing the total weighted interferences (with a weight of $P_{n}$) among all the APs.

\subsection{\label{sec:Distributed-Channel-Selection}Non-Cooperative AP Channel
Selection Algorithm}

The purpose of designing this algorithm  is to allow APs to select their
channels in a distributed manner to achieve a mutually acceptable
resource allocation, i.e., an Nash equilibrium. The key idea is to
let APs asynchronously improve their channel selections according
to the finite improvement property.

\begin{algorithm}[tt]
\begin{algorithmic}[1]
\State \textbf{initialization:}
\State \hspace{0.4cm} \textbf{set} the initial channel $a_{n}(0)=m_{n}$ for
each AP $n\in\mathcal{N}$, initial channel selection profile as $\boldsymbol{a}(0)=(a_{1}(0),...,a_{N}(0))$, and the stage index $t=0$.
\State \textbf{end initialization\newline}

\While{$\boldsymbol{a}(t)$ is not a Nash equilibrium}
        \For{AP $n=1$ to $N$}
            \State \textbf{choose} the channel $a_{n}(t+1)$ that maximizes its own throughput according to (\ref{eq:bs1}).
        \EndFor
        \State \textbf{set} channel selection profile as $\boldsymbol{a}(t+1)=(a_{1}(t+1),...,a_{N}(t+1))$ and the stage index $t=t+1$.
\EndWhile

\end{algorithmic}
\caption{\label{alg:Online-Distributed-Channel} Non-Cooperative AP Channel Selection Algorithm}
\end{algorithm}

We assume that when an AP queries the geo-location database, the database
will assign it with a unique ID indexed as $1,2,3,...$. For initialization,
we let each AP $n$ select the channel $m_{n}$ that has the smallest channel ID index among its feasible channels $\mathcal{M}_{n}$, i.e., $a_{n}(0)=m_{n}$.
Then based on the initialized channel selection profile $\boldsymbol{a}(0)=(a_{1}(0),...,a_{N}(0))$,
each AP $n$ in turn (according to the assigned IDs) carries out the
best response update, i.e., select a channel $a_{n}(t+1)$ that maximizes
its own throughput as \begin{eqnarray}
a_{n}(t+1) = \arg\max_{a\in\mathcal{M}_{n}}U_{n}(a,a_{1}(t+1),...,a_{n-1}(t+1),a_{n+1}(t),...,a_{N}(t)),\label{eq:bs1}\end{eqnarray}
given the channel selections $\{a_{1}(t+1),...,a_{n-1}(t+1)\}$ of
the updated APs, and the channel selections $\{a_{n+1}(t),...,a_{N}(t)\}$
of remaining APs that are not updated at the current stage $t$. Such update procedure
continues until a Nash equilibrium is reached. Since the best
response update is also a better response update, according to the
finite improvement property, such asynchronous best response updates
must achieve a Nash equilibrium within finite number of iterations.
We summarize the non-cooperative AP channel selection algorithm in Algorithm
\ref{alg:Online-Distributed-Channel}. We then consider the computational complexity of the algorithm. Lines $5$ to $7$ involves $N$ maximization operations and each maximization operation can be achieved by sorting over at most $M$ values. This step typically has a complexity of $\mathcal{O}(NM\log M)$. Line $10$ has the complexity of $\mathcal{O}(1)$. Suppose that it takes $C$ iterations for the algorithm to converge. Then total computational complexity of the algorithm is $\mathcal{O}(CNM\log M)$. Similarly, the space complexity of the algorithm is $\mathcal{O}(NM)$.

The Algorithm \ref{alg:Online-Distributed-Channel} requires all
APs to truthfully communicate with each other about their channel selections. When
such a requirement is not feasible, each AP can independently implement
Algorithm \ref{alg:Online-Distributed-Channel} by acquiring the assigned
IDs, available channels, and transmission powers of other APs from
the database. Note that such an off-line implementation is incentive
compatible, since given other APs adhere to the algorithm and the update order is fixed, no AP
has an incentive to deviate unilaterally from the algorithm (due to
the deterministic Nash equilibrium output).

\subsection{\label{sec:PoA}Price of Anarchy}
We now study the efficiency of Nash equilibria of the non-cooperative AP channel selection Game. Following the definition of price of anarchy (PoA) in game theory \cite{roughgarden2007introduction}, we will quantify the efficiency ratio of the worst-case Nash equilibrium over the optimal solution by the cooperative AP channel selection. Let $\Xi$ be the set of Nash equilibria of the game. Then the PoA is defined as\[
\mbox{PoA}=\frac{\min_{\boldsymbol{a}\in\Xi}\sum_{n=1}^{N}U_{n}(\boldsymbol{a})}{\max_{\boldsymbol{a}\in\prod_{n=1}^{N}\mathcal{M}_{n}}\sum_{n=1}^{N}U_{n}(\boldsymbol{a})},\]
which is always not greater than $1$. A larger PoA implies that the
set of Nash equilibrium is more efficient (in the worst-case sense when comparing with the system optimal solution). Let $\overline{\omega_{n}}=\max_{m\in\mathcal{M}_{n}}\{\omega_{m}^{n}\}$
and $\underline{\omega_{n}}=\min_{m\in\mathcal{M}_{n}}\{\omega_{m}^{n}\}$.
We can first show that
\begin{lem}
\label{lem:PoA}For the non-cooperative AP channel selection game,
the throughput of an AP $n\in\mathcal{N}$ at a Nash equilibrium is
no less than $B\log_{2}\left(1+\frac{P_{n}/d_{n}^{\theta}}{\overline{\omega_{n}}+\left(\sum_{i\in\mathcal{N}/\{n\}}P_{i}/d_{in}^{\theta}\right)/|\mathcal{M}_{n}|}\right),$
where $|\mathcal{M}_{n}|$ is the number of vacant channels for AP
$n$.\end{lem}
\begin{proof}
We will prove the result by contradiction. Suppose that an AP $n$
at a Nash equilibrium $\boldsymbol{a}^{*}$ has a throughput less
than $B\log_{2}\left(1+\frac{P_{n}/d_{n}^{\theta}}{\overline{\omega_{n}}+\left(\sum_{i\in\mathcal{N}/\{n\}}P_{i}/d_{in}^{\theta}\right)/|\mathcal{M}_{n}|}\right)$.
From the throughput function in (\ref{eq:U}), we must have that \begin{equation}
\omega_{a_{n}^{*}}^{n}+\sum_{i\in\mathcal{N}/\{n\}:a_{i}^{*}=a_{n}^{*}}P_{i}/d_{in}^{\theta}>\overline{\omega_{n}}+\left(\sum_{i\in\mathcal{N}/\{n\}}P_{i}/d_{in}^{\theta}\right)/|\mathcal{M}_{n}|.\label{eq:poa1}\end{equation}
According to the definition of Nash equilibrium (no AP can improve
by changing channel unilaterally), we also have that\begin{equation}
\omega_{m}^{n}+\sum_{i\in\mathcal{N}/\{n\}:a_{i}^{*}=m}P_{i}/d_{in}^{\theta}\geq\omega_{a_{n}^{*}}^{n}+\sum_{i\in\mathcal{N}/\{n\}:a_{i}^{*}=a_{n}^{*}}P_{i}/d_{in}^{\theta},\forall m\in\mathcal{M}_{n},\label{eq:poa2}\end{equation}
which implies that\begin{eqnarray}
 &  & \sum_{m\in\mathcal{M}_{n}}\left(\omega_{m}^{n}+\sum_{i\in\mathcal{N}/\{n\}:a_{i}^{*}=m}P_{i}/d_{in}^{\theta}\right)\nonumber \\
 & = & \sum_{m\in\mathcal{M}_{n}}\omega_{m}^{n}+\sum_{i\in\mathcal{N}/\{n\}}P_{i}/d_{in}^{\theta}\geq|\mathcal{M}_{n}|\left(\omega_{a_{n}^{*}}^{n}+\sum_{i\in\mathcal{N}/\{n\}:a_{i}^{*}=a_{n}^{*}}P_{i}/d_{in}^{\theta}\right).\label{eq:poa3}\end{eqnarray}
According to (\ref{eq:poa1}) and (\ref{eq:poa3}), we now reach a
contradiction that \begin{equation*}
\left(\sum_{m\in\mathcal{M}_{n}}\omega_{m}^{n}\right)/|\mathcal{M}_{n}|+\left(\sum_{i\in\mathcal{N}/\{n\}}P_{i}/d_{in}^{\theta}\right)/|\mathcal{M}_{n}|
  >  \overline{\omega_{n}}+\left(\sum_{i\in\mathcal{N}/\{n\}}P_{i}/d_{in}^{\theta}\right)/|\mathcal{M}_{n}|.\end{equation*}
This proves the result.
\end{proof}
Lemma \ref{lem:PoA} implies that at a Nash equilibrium each AP will
receive an interference level that is not greater than the maximum
possible interference level (i.e., $\sum_{i\in\mathcal{N}/\{n\}}P_{i}/d_{in}^{\theta}$)
divided by the number of its available channels. That is, if more
channels are available then the performance of Nash equilibria can
be improved. According to Lemma \ref{lem:PoA}, we know that
\begin{cor}
The PoA of the non-cooperative AP channel selection game is lower bounded by \[\frac{\sum_{n=1}^{N}\log_{2}\left(1+\frac{P_{n}/d_{n}^{\theta}}{\overline{\omega_{n}}+\left(\sum_{i\in\mathcal{N}/\{n\}}P_{i}/d_{in}^{\theta}\right)/|\mathcal{M}_{n}|}\right)}{\sum_{n=1}^{N}\log_{2}\left(1+\frac{P_{n}/d_{n}^{\theta}}{\underline{\omega_{n}}}\right)}.\]\end{cor}
\begin{proof}
According to Lemma \ref{lem:PoA}, we have that\begin{eqnarray*}
\mbox{PoA} & \geq & \frac{\sum_{n=1}^{N}B\log_{2}\left(1+\frac{P_{n}/d_{n}^{\theta}}{\overline{\omega_{n}}+\left(\sum_{i\in\mathcal{N}/\{n\}}P_{i}/d_{in}^{\theta}\right)/|\mathcal{M}_{n}|}\right)}{\max_{\boldsymbol{a}\in\prod_{n=1}^{N}\mathcal{M}_{n}}\sum_{n=1}^{N}B\log_{2}\left(1+\frac{P_{n}/d_{n}^{\theta}}{\omega_{a_{n}}^{n}+\sum_{i\in\mathcal{N}/\{n\}:a_{i}=a_{n}}P_{i}/d_{in}^{\theta}}\right)}\\
 & > & \frac{\sum_{n=1}^{N}\log_{2}\left(1+\frac{P_{n}/d_{n}^{\theta}}{\overline{\omega_{n}}+\left(\sum_{i\in\mathcal{N}/\{n\}}P_{i}/d_{in}^{\theta}\right)/|\mathcal{M}_{n}|}\right)}{\sum_{n=1}^{N}\log_{2}\left(1+\frac{P_{n}/d_{n}^{\theta}}{\underline{\omega_{n}}}\right)}.\end{eqnarray*}
\end{proof}

The PoA characterizes the worst-case performance of Nash equilibria.
Numerical results in Section VII demonstrate that the convergent Nash
equilibrium of the proposed algorithm in Section \ref{sec:Distributed-Channel-Selection} is often more efficient than what the PoA indicates
and the performance loss is less than $8\%$, compared with the optimal
solution by the cooperative AP channel selection.

%% file: User.tex
\section{\label{sec:Distributed-AP-Association-1}Distributed AP Association
By Mobile Secondary Users}

We now consider the distributed AP association problem among a set of
mobile secondary users $\mathcal{K}=\{1,2,...,K\}$ in the SU tier. Let $x_{n}$
be the number of users that associate with AP $n$, which satisfies
that $\sum_{n=1}^{N}x_{n}=K$. We assume that the APs' cooperative/non-cooperative channel selections in the AP tier and the users' AP associations in the SU tier are decoupled, i.e., APs only interested in guaranteeing their throughputs by proper channel selections and users can improve their data rates by proper AP associations. The load-aware AP channel selection will be considered in a future work.

\subsection{Channel Contention Within an AP}

We first consider the channel contention when multiple secondary users
associate with the same AP. Here we adopt a random backoff mechanism to resolve the channel contention. More specifically, the time is slotted (see
Figure \ref{fig:Time-Slot-Structure}), with a contention stage being
divided into $\lambda_{\max}$ mini-slots.\footnote{For the ease of exposition, we assume that the contention backoff size $\lambda_{\max}$ is fixed. This corresponds to an equilibrium model for the case that the backoff size $\lambda_{\max}$ can be dynamically tuned according to the 802.11 distributed coordination
function \cite{bianchi2000performance}. Also, we can enhance the performance of the backoff mechanism by determining optimal fixed contention backoff size according to the method in \cite{kriminger2011markov}.} Each secondary user $k$
executes the following two steps:
\begin{enumerate}
\item Count down according to a randomly and uniformly chosen integral backoff
time (number of mini-slots) $\lambda_{k}$ between $1$ and $\lambda_{\max}$.
\item Once the timer expires, monitor the channel and exchange RTS/CTS messages
with the AP in order to grab the channel if the channel is clear (i.e.,
no ongoing transmission). Note that if multiple users choose the same
backoff mini-slot, a collision will occur with RTS/CTS transmissions
and no users can grab the channel. Once the RTS/CTS message exchange
goes through, then the AP starts to transmit the data packets to the
user.
\end{enumerate}

Since $x_{n}$ users contend for the channel in AP $n$, the probability
that a user $k$ (out of these $x_{n}$ users) grabs the channel successfully is\begin{align}
g(x_{n}) = Pr\{\lambda_{k}<\min_{i\neq k}\{\lambda_{i}\}\} = \sum_{\lambda=1}^{\lambda_{\max}}Pr\{\lambda_{k}=\lambda\}Pr\{\lambda<\min_{i\neq k}\{\lambda_{i}\}|\lambda_{k}=\lambda\} = \sum_{\lambda=1}^{\lambda_{\max}}\frac{1}{\lambda_{\max}}\left(\frac{\lambda_{\max}-\lambda}{\lambda_{\max}}\right)^{x_{n}-1},\label{eq:gg}\end{align}
which is a decreasing function of the total number of contending users
$x_{n}$. Then the average data rate of a secondary user $k$ associating
with AP $n$ is given as\begin{equation}
r_{k}=H_{n}^{k}U_{n}(\boldsymbol{a}^{*})g(x_{n}),\label{eq:u1}\end{equation}
where $U_{n}(\boldsymbol{a}^{*})$ is the throughput at the boundary of the coverage area of AP $n$ at the equilibrium
channel selections $\boldsymbol{a}^{*}$ by cooperative/non-cooperative AP channel selection algorithms, and $H_{n}^{k}\geq1$ is the transmission gain of user $k$. Here the transmission gain is used to model user specific throughputs due to their heterogeneous channel conditions. For example, a user enjoys a better channel condition than all other users if it is the closest to the AP.

\begin{figure}[tt]
\centering
\includegraphics[scale=0.7]{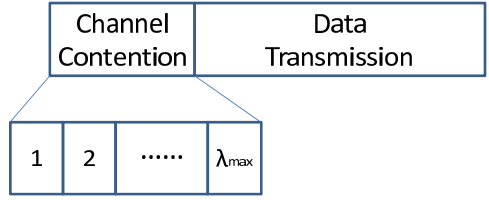}
\caption{\label{fig:Time-Slot-Structure}Time slot structure of channel contention}
\end{figure}

\subsection{Distributed AP Association Game}

Due to the channel contention
within an AP $n$, the average data rate $r_{k}$ of a secondary user
$k$ decreases with the total number of contending users $x_{n}$.
To improve the data rate $r_{k}$, the secondary user $k$ can choose
to move to another AP $n'$ with less users. However, in
practices people may not prefer long distance movements (just for the sake of obtaining better communication experiences), which motivates
us to take the cost of mobility into account. By defining the current
location profile of all secondary users as a system state, we next
formulate the distributed AP association problem as a \textbf{state-based
game} \cite{LiM10} as follows:
\begin{itemize}
\item Player $k$: a secondary user from the set $\mathcal{K}$.
\item Strategy $b_{k}$: choose an AP $n\in\mathcal{N}$ to associate with. We denote the strategy profile of all
users as $\boldsymbol{b}\triangleq(b_{1},...,b_{K})$.
\item State $\boldsymbol{s}\triangleq(s_{1},...,s_{K})$: the current locations (i.e.,
the associated APs) of all secondary users, where $s_{k}$ denote
the location of user $k$.
\item State Transition $\boldsymbol{s}^{'}=F(\boldsymbol{b},\boldsymbol{s})$: in general the new state $\boldsymbol{s}^{'}$ is determined by the strategies $\boldsymbol{b}$ of all secondary users and the original state $\boldsymbol{s}$, where
$F(\cdot)$ denotes the state transition function. For our problem, we
have that $F(\boldsymbol{b},\boldsymbol{s})=\boldsymbol{b}$, i.e., the new locations just depend on secondary users' AP choices and independent of the original system state.
\item Payoff $V_{k}(\boldsymbol{b},\boldsymbol{s})$: secondary user $k$'s  utility obtained from
the strategy profile $\boldsymbol{b}$ in state $\boldsymbol{s}$. To take the cost of mobility
into account, we define \begin{eqnarray}
V_{k}(\boldsymbol{b},\boldsymbol{s})  =  r_{k}-\delta_{k}d_{b_{k}s_{k}} = H_{b_{k}}^{k}U_{b_{k}}(\boldsymbol{a}^{*})g(x_{b_{k}}(\boldsymbol{b}))-\delta_{k}d_{b_{k}s_{k}},\label{eq:V1}\end{eqnarray}
where $x_{b_{k}}(\boldsymbol{b})$ is the number of contending users associated with AP
$b_{k}$ under strategy profile $\boldsymbol{b}$, $\delta_{k}>0$ is the factor
representing the weight of mobility cost in user $k$'s decision,
and $d_{b_{k}s_{k}}$ is the distance of moving to AP $b_{k}$ from AP $s_{k}$ ($d_{b_{k}s_{k}}=0$ if $b_{k}=s_{k}$).  Note that the distance measure here can represent more general preference functions and can also be asymmetric. For example, we can define that $d_{b_{k}^{'}b_{k}}>d_{b_{k}b_{k}^{'}}$ if $b_{k}$ is a popular shopping mall where uses like to stay. The physical meaning
of (\ref{eq:V1}) is to balance the average data rate that a user
can obtain from moving to a new AP $b_{k}$ with the mobility cost by
moving from its current associated AP $s_{k}$.
\end{itemize}

Since the state-based game is a generalized game theoretic framework
(by regarding the classical strategic game as a state-based game with
a constant state), we need an updated equilibrium concept. Here
we follow the recent results in \cite{LiM10} and introduce the state-based
Nash equilibrium. To proceed, we first define the set of reachable
states $\triangle(\boldsymbol{b}^{0},\boldsymbol{s}^{0})$ starting from a strategy state pair
$(\boldsymbol{b}^{0},\boldsymbol{s}^{0})$ as \begin{equation}
\triangle(\boldsymbol{b}^{0},\boldsymbol{s}^{0})\triangleq\{\boldsymbol{s}^{t}:\boldsymbol{s}^{t}=F(\boldsymbol{b}^{0},\boldsymbol{s}^{t-1}),\forall t\geq1\}.\label{eq:s1}\end{equation}
We then extend the definition of Nash equilibrium to the state-based
game setting as follows.
\begin{defn}[\textbf{State-based Nash Equilibrium} \!\!\cite{LiM10}]
\label{def:A-strategy-state}A strategy state pair $(\boldsymbol{b}^{*},\boldsymbol{s}^{*})$
is a state-based Nash equilibrium if \\
1) the state $\boldsymbol{s}^{*}$ is reachable from $(\boldsymbol{b}^{*},\boldsymbol{s}^{*})$, i.e.,
$\boldsymbol{s}^{*}\in\triangle(\boldsymbol{b}^{*},\boldsymbol{s}^{*})$.\\
2) for every player $k\in\mathcal{K}$ and every state $\boldsymbol{s}\in\triangle(\boldsymbol{b}^{*},\boldsymbol{s}^{*})$,
we have\begin{equation}
V_{k}(\boldsymbol{b}^{*},\boldsymbol{s})=\max_{b_{k}}V_{k}(b_{k},b_{-k}^{*},\boldsymbol{s}).\label{eq:s2}\end{equation}\end{defn}

The physical meaning of the state-based Nash equilibrium is that the
state $\boldsymbol{s}^{*}$ is recurrent and the strategy profile $\boldsymbol{b}^{*}$ is the
best response no matter how the game state evolves after-wards. In
principle, the state-based game is a special case of the stochastic
game, which is difficult to tackle. However, we are able to solve
the distributed AP association game by exploiting its inherent structure
property. A key observation is that, similarly to the classical
potential game, the state-based distributed AP association game also
admits a state-based potential function as\begin{equation}
\Psi(\boldsymbol{b},\boldsymbol{s})=\sum_{k=1}^{K}\ln U_{b_{k}}(\boldsymbol{a}^{*})+\sum_{n=1}^{N}\sum_{i=0}^{x_{n}(\boldsymbol{b})}\ln g(i)+\sum_{k=1}^{K}\ln H_{b_{k}}^{k}\label{eq:s3}\end{equation}

For the state-based potential function $\Psi(\boldsymbol{b},\boldsymbol{s})$, we have
\begin{lem}\label{lemmaS}
For the state-based distributed AP association game, if a player $k\in\mathcal{K}$
performs a better response update $b_{k}$ in a given state $\boldsymbol{s}=(s_{k},s_{-k})$
with \[
V_{k}(b_{k},s_{-k},\boldsymbol{s})>V_{k}(\boldsymbol{s},\boldsymbol{s}),\]
we then have that\[
\Psi(b_{k},s_{-k},\boldsymbol{s})>\Psi(\boldsymbol{s},\boldsymbol{s}).\]
\end{lem}
\begin{proof}
First of all, the condition $V_{k}(b_{k},s_{-k},\boldsymbol{s})>V_{k}(\boldsymbol{s},\boldsymbol{s})$
implies that\[
H_{b_{k}}^{k}U_{b_{k}}(\boldsymbol{a}^{*})g(x_{b_{k}}(b_{k},s_{-k}))-\delta_{k}d_{b_{k}s_{k}}>H_{s_{k}}^{k}U_{s_{k}}(\boldsymbol{a}^{*})g(x_{s_{k}}(s_{k},s_{-k}))-\delta_{k}d_{s_{k}s_{k}}.\]
Since $d_{b_{k}s_{k}}>d_{s_{k}s_{k}}=0$, we then have\begin{equation}
H_{b_{k}}^{k}U_{b_{k}}(\boldsymbol{a}^{*})g(x_{b_{k}}(b_{k},s_{-k}))>H_{s_{k}}^{k}U_{s_{k}}(\boldsymbol{a}^{*})g(x_{s_{k}}(s_{k},s_{-k})).\label{eq:lll1}\end{equation}
Second, since $x_{b_{k}}(b_{k},s_{-k})=x_{b_{k}}(s_{k},s_{-k})+1$
and $x_{s_{k}}(s_{k},s_{-k})=x_{s_{k}}(b_{k},s_{-k})+1$, we have
that\begin{eqnarray}
 &  & \Psi(b_{k},s_{-k},\boldsymbol{s})-\Psi(s_{k},s_{-k},\boldsymbol{s})\nonumber \\
 & = & \ln U_{b_{k}}(\boldsymbol{a}^{*})-\ln U_{s_{k}}(\boldsymbol{a}^{*})+\sum_{i=0}^{x_{b_{k}}(b_{k},s_{-k})}\ln g(i)-\sum_{i=0}^{x_{b_{k}}(s_{k},s_{-k})}\ln g(i)\nonumber \\
 &  & +\sum_{i=0}^{x_{s_{k}}(b_{k},s_{-k})}\ln g(i)-\sum_{i=0}^{x_{s_{k}}(s_{k},s_{-k})}\ln g(i)+\ln H_{b_{k}}^{k}-\ln H_{s_{k}}^{k}\nonumber \\
 & = & \ln U_{b_{k}}(\boldsymbol{a}^{*})-\ln U_{s_{k}}(\boldsymbol{a}^{*})+\ln g(x_{b_{k}}(b_{k},s_{-k}))-\ln g(x_{s_{k}}(s_{k},s_{-k}))+\ln H_{b_{k}}^{k}-\ln H_{s_{k}}^{k}\nonumber \\
 & = & \ln\left(H_{b_{k}}^{k}U_{b_{k}}(\boldsymbol{a}^{*})g(x_{b_{k}}(b_{k},s_{-k}))\right)-\ln\left(H_{s_{k}}^{k}U_{s_{k}}(\boldsymbol{a}^{*})g(x_{s_{k}}(s_{k},s_{-k}))\right).\label{eq:llll2}\end{eqnarray}
From (\ref{eq:lll1}) and (\ref{eq:llll2}), we must have that $\Psi(b_{k},s_{-k},\boldsymbol{s})>\Psi(s_{k},s_{-k},\boldsymbol{s})$.
\end{proof}

Similarly to  the classical potential game, we can also define the
finite improvement property for the state-based game. Let $\boldsymbol{s}^{t}=(s_{1}^{t},...s_{K}^{t})$
be the state of the game in the $t$-th update, and $\boldsymbol{b}^{t}=(b_{1}^{t},...b_{K}^{t})$
be the strategy profile of all players in $t$-th update. According
to the state transition, we have $\boldsymbol{s}^{t+1}=F(\boldsymbol{b}^{t},\boldsymbol{s}^{t})$. A path
of the state-based game is a sequence $\rho=((\boldsymbol{b}^{0},\boldsymbol{s}^{0}),(\boldsymbol{b}^{1},\boldsymbol{s}^{1}),...)$
such that for every $t\geq1$ there exists a unique player, say player
$k_{t}$, such that $\boldsymbol{b}^{t}=(b_{k_{t}},s_{-k}^{t})$ for some strategy
$b_{k_{t}}\neq s_{k_{t}}^{t}$. $\rho=((\boldsymbol{b}^{0},\boldsymbol{s}^{0}),(\boldsymbol{b}^{1},\boldsymbol{s}^{1}),...)$
is an improvement path if for all $t\geq1$ we have $V_{k_{t}}(\boldsymbol{b}^{t},\boldsymbol{s}^{t})>V_{k_{t}}(\boldsymbol{s}^{t},\boldsymbol{s}^{t})$,
where $k_{t}$ is the unique deviator at the $t$-th update. From the properties
of the state-based potential function $\Psi(\boldsymbol{b},\boldsymbol{s})$, we first show
that every improvement path is finite.
\begin{thm}
\label{thm:The-distributed-AP}For the state-based distributed AP
association game, every improvement path is finite.\end{thm}
\begin{proof}
For any improvement path $\rho=((\boldsymbol{b}^{0},\boldsymbol{s}^{0}),(\boldsymbol{b}^{1},\boldsymbol{s}^{1}),...)$,
we have \begin{eqnarray*}
V_{k_{0}}(\boldsymbol{b}^{0},\boldsymbol{s}^{0}) & > & V_{k_{0}}(\boldsymbol{s}^{0},\boldsymbol{s}^{0}),\\
V_{k_{1}}(\boldsymbol{b}^{1},\boldsymbol{s}^{1}) & > & V_{k_{1}}(\boldsymbol{b}^{0},\boldsymbol{s}^{1}),\\
V_{k_{2}}(\boldsymbol{b}^{2},\boldsymbol{s}^{2}) & > & V_{k_{2}}(\boldsymbol{b}^{1},\boldsymbol{s}^{2}),\\
 & \vdots\end{eqnarray*}
where $\boldsymbol{s}^{1}=F(\boldsymbol{b}^{0},\boldsymbol{s}^{0})=\boldsymbol{b}^{0}$, $\boldsymbol{s}^{2}=F(\boldsymbol{b}^{1},\boldsymbol{s}^{1})=\boldsymbol{b}^{1}$, and so on. From Lemma \ref{lemmaS}, we know
that\begin{eqnarray*}
\Psi(\boldsymbol{s}^{0},\boldsymbol{s}^{0}) & < & \Psi(\boldsymbol{b}^{0},\boldsymbol{s}^{0}),\\
\Psi(\boldsymbol{b}^{0},\boldsymbol{s}^{0}) & < & \Psi(\boldsymbol{b}^{0},\boldsymbol{s}^{1}),\\
\Psi(\boldsymbol{b}^{0},\boldsymbol{s}^{1}) & < & \Psi(\boldsymbol{b}^{1},\boldsymbol{s}^{1}),\\
\Psi(\boldsymbol{b}^{1},\boldsymbol{s}^{1}) & < & \Psi(\boldsymbol{b}^{1},\boldsymbol{s}^{2}),\\
\Psi(\boldsymbol{b}^{1},\boldsymbol{s}^{2}) & < & \Psi(\boldsymbol{b}^{2},\boldsymbol{s}^{2}),\\
 & \vdots\end{eqnarray*}
which is increasing along the improvement path. Since $\Psi(\boldsymbol{b},\boldsymbol{s})<\infty$,
then the improvement path $\rho=((\boldsymbol{b}^{0},\boldsymbol{s}^{0}),(\boldsymbol{b}^{1},\boldsymbol{s}^{1}),...)$
must be finite.
\end{proof}
Similarly to the classical potential game, we further show that
any asynchronous better response update process also leads to a state-based
Nash equilibrium.
\begin{thm}
\label{thm:For-the-distributed}For the state-based distributed AP
association game, any asynchronous better response update process leads to
a state-based Nash equilibrium $(\boldsymbol{b}^{*},\boldsymbol{s}^{*})$ with $\boldsymbol{s}^{*}=F(\boldsymbol{b}^{*},\boldsymbol{s}^{*})$.\end{thm}
\begin{proof}
Suppose that an asynchronous better response update process $\rho=((\boldsymbol{b}^{0},\boldsymbol{s}^{0}),(\boldsymbol{b}^{1},\boldsymbol{s}^{1}),...)$
terminates at the point $(\boldsymbol{b}^{*},\boldsymbol{s}^{*})$. In this case, we must have
$\Psi(b_{k}^{*},b_{-k}^{*},\boldsymbol{s}^{*})\geq\max_{b_{k}}\Psi(b_{k},b_{-k}^{*},\boldsymbol{s}^{*})$, otherwise the improvement path $\rho$ does not terminate at point
$(\boldsymbol{b}^{*},\boldsymbol{s}^{*})$. At point
$(\boldsymbol{b}^{*},\boldsymbol{s}^{*})$, we must also have that $V_{k}(b_{k}^{*},b_{-k}^{*},\boldsymbol{s}^{*})\geq\max_{b_{k}}V_{k}(b_{k},b_{-k}^{*},\boldsymbol{s}^{*})$ and $\boldsymbol{s}^{*}=\boldsymbol{b}^{*}=F(\boldsymbol{b}^{*},\boldsymbol{s}^{*})$, otherwise the
potential function can be improved and thus the improvement path does not terminate here. Thus, we have $\boldsymbol{s}^{*}\in\triangle(\boldsymbol{b}^{*},\boldsymbol{s}^{*})$
and $V_{k}(b_{n}^{*},b_{-n}^{*},\boldsymbol{s})\geq\max_{b_{n}}V_{k}(b_{n},b_{-n}^{*},\boldsymbol{s}),\forall \boldsymbol{s}\in\triangle(\boldsymbol{b}^{*},\boldsymbol{s}^{*})$,
which satisfies the conditions in Definition \ref{def:A-strategy-state}.
\end{proof}

Since $\boldsymbol{s}^{*}=\boldsymbol{b}^{*}=F(\boldsymbol{b}^{*},\boldsymbol{s}^{*})$, Theorem \ref{thm:For-the-distributed} implies that the asynchronous better response update process leads to the state-based Nash equilibrium ($\boldsymbol{b}^{*},\boldsymbol{b}^{*}$), i.e., the equilibrium that all users are satisfied with the current AP associations $\boldsymbol{b}^{*}$ and have no incentive to move anymore.

\subsection{\label{sec:Distributed-AP-Association}Distributed AP Association
Algorithm}

We next design a distributed AP association algorithm based on the
finite improvement property shown in Theorem \ref{thm:The-distributed-AP},
which allows secondary users to select their associated APs in a distributed
manner and achieve mutually acceptable AP associations, i.e., a state-based Nash
equilibrium.

\begin{algorithm}[tt]
\begin{algorithmic}[1]
\State \textbf{initialization:}
\State \hspace{0.4cm} \textbf{set} the mean $\eta$ for strategy update countdown.
\State \textbf{end initialization\newline}

\Loop{ for each secondary user $k\in\mathcal{K}$ in parallel:}
        \State \textbf{generate}  a timer value that follows the exponential distribution
with the mean $\eta$.
        \State \textbf{count down} until the timer expires.
        \If{ the timer expires}
            \State \textbf{acquire} the information of channel throughput $\{U_{n}(\boldsymbol{a}^{*})\}$, the geo-location of APs, and user distribution $\{x_{n}\}_{n=1}^{N}$.
            \State \textbf{update} the strategy $b_{k}^{*}$ according to the best response in (\ref{eq:ss0}).
        \EndIf
\EndLoop
\end{algorithmic}
\caption{\label{alg:-Distributed-AP} Distributed AP Association Algorithm}
\end{algorithm}

The key idea is to let secondary users asynchronously improve their
AP selections. Unlike the non-cooperative AP channel selection update
with the fixed order enforced by the geo-location database, the distributed AP
association algorithm can not be deterministic. This is because that, as secondary users dynamically enter and leave the network, a deterministic distributed AP association
algorithm according to the fixed strategy update order is not robust.
Hence we will design a randomized algorithm by letting each
secondary user countdown according to a timer value that follows the
exponential distribution with a mean equal to $\eta$. Since the exponential distribution has support over $(0,\infty)$ and its probability density function is continuous,
the probability that more than one users generate the same timer value
and update their strategies simultaneously equals zero.\footnote{The timer in practice is always finite, and the collision probability is not exactly zero. However, as long as the collision probability is very small, the following analysis is a very good approximation of the reality.} When a user
$k$ activates its strategy update at time $t$, the user can computes
its best response strategy as \begin{eqnarray}
b_{k}^{*} & = & \arg\max_{b_{k}}V_{k}(b_{k},b_{-k}^{t},\boldsymbol{s}^{t}) = \arg\max_{b_{k}}H_{b_{k}}^{k}U_{b_{k}}(\boldsymbol{a}^{*})g(x_{b_{k}}(\boldsymbol{b}^{t}))-\delta_{k}d_{b_{k}s_{k}},\label{eq:ss0}\end{eqnarray}
which requires the information of user distribution $(x_{1}(\boldsymbol{b}^{t}),...x_{N}(\boldsymbol{b}^{t}))$
at time $t$, the throughput $U_{n}(\boldsymbol{a}^{*})$, and geo-locations of
all the APs.  We then consider the computational complexity of the algorithm. For each iteration of each user, Lines $4$ to $7$ only involve random value generation and subduction operation for
count-down, and hence have a complexity of $\mathcal{O}(1)$. Line $8$ involves information inquiry from $N$ APs and hence has a complexity of $\mathcal{O}(N)$. Line $9$ computes  the best response strategy, which can be achieved by sorting at most $N$ values and typically has a complexity of $\mathcal{O}(N\log N)$. Suppose that it takes $C$ iterations for the algorithm
to converge. Then total computational complexity of $K$ users is $\mathcal{O}(CKN\log N)$. Similarly, we can show that the space complexity is $\mathcal{O}(KN)$.

To facilitate the best response update, we propose to
setup a social database (accessible by all secondary users), wherein
each AP reports its channel throughput $U_{n}(\boldsymbol{a}^{*})$ and geo-location,
and each secondary user $k^{'}\in\mathcal{K}$ posts and shares its
AP association $b_{k^{'}}^{*}$ with other users in the manner like
Twitter once it updates. Based on the information from the social
database, a secondary user $k$ can first figure out the user distribution
$(x_{1}(\boldsymbol{b}^{t}),...x_{N}(\boldsymbol{b}^{t}))$ as \begin{equation}
x_{n}(\boldsymbol{b}^{t})=\sum_{k^{'}=1}^{K}I_{\{b_{k'}^{t}=n\}},\forall n\in\mathcal{N},\label{eq:ss1}\end{equation}
where $I_{\{b_{k'}^{t}=n\}}=1$ if user $k'$ associates with AP $n$,
and $I_{\{b_{k'}^{t}=n\}}=0$ otherwise. Based on the user distribution,
the secondary user $k$ can then compute the corresponding best response
strategy according to (\ref{eq:ss0}).

The success of social database requires that each user is willing
to share the information of its AP association. When this
is not feasible, each AP $n$ can estimate its associated user population $x_{n}$
locally. Let $\bar{g}(x_{n})$ denote the probability that no user
among $x_{n}$  associated with the same AP grabs the channel in a time slot
$\tau$. This can be computed as $\bar{g}(x_{n})=1-x_{n}g(x_{n})$,
where $g(x_{n})$ is given in (\ref{eq:gg}). In a time slot $\tau$, AP $n$
can observe the information $I_{a_{n}^{*}}(\tau)\in\{1,0\}$, i.e., whether the channel $a_{n}^{*}$ is used by any users or not. Then
over a long period that consists of $L$ time slots, AP $n$ can observe
the outcome $\{I_{a_{n}^{*}}(\tau)\}_{\tau=1}^{L}$ and estimate $\bar{g}(x_{n})=\frac{\sum_{\tau=1}^{L}I_{a_{n}^{*}}(\tau)}{L}$
by the sample-average. Since $I_{a_{n}^{*}}(\tau)$ is independently
and identically distributed according to the probability $\bar{g}(x_{n})$,
according to the law of large numbers, the estimation will be accurate when the observation period length $L$ is large enough.
This is feasible in practices since user's mobility decision is often
carried out at a large time scale (say every few minutes), compared
with the time scale of a time slot (say $50$ microseconds in the
standard 802.11 system). Then AP $n$ can obtain the number of its
associated users $x_{n}$ by solving that $x_{n}=\bar{g}^{-1}\left(\frac{\sum_{\tau=1}^{L}I_{a_{n}^{*}}(\tau)}{L}\right)$, and report it in the social database.

We summarize the distributed
AP association algorithm in Algorithm \ref{alg:-Distributed-AP}.
According to Theorem \ref{thm:For-the-distributed}, such asynchronous
best response update process must reach a state-based Nash equilibrium.
Numerical results show that the algorithm is also robust to the dynamics
of secondary users' leaving and entering the system.

%% file: Sim.tex
\section{\label{sec:Simulation-Results}Simulation Results}

In this part, we investigate the proposed AP channel selection
and AP association algorithms by simulations.

\subsection{\label{sub:Distributed-AP-Channel}Cooperative AP Channel Selection}

We first implement the cooperative AP channel selection algorithm
in Section \ref{sec:Distributed-AP-Channel1}. We consider a
white-space wireless system consisting of $M=4$ channels and $N=8$
APs, which are scattered across a square area of a length of $500$
m (see Figure \ref{fig:A-square-area}). The bandwidth of each channel
is $6$ MHz, the noise power is $\omega_{m}^{n}=-100$ dBm, and the
path loss factor $\theta=4$. Each AP $n$ operates with a specific
transmission power $P_{n}$ and has a different set of vacant channels
by consulting the geo-location database (please refer to Figure \ref{fig:A-square-area}
for the details of these parameters). We set that the distance $d_{n}$
between AP $n$ and its associated boundary secondary user is $20$ m.

We implement the cooperative channel selection algorithm with the parameter $\gamma=0.2$, $0.5,$ and $0.85$, respectively\footnote{Note that the system throughput $\sum_{i=1}^{N}U_{i}(\boldsymbol{a})$ is a large number in the Mbps unit and a small $\gamma$ is adopted in the simulation. Otherwise, $\exp(\gamma\sum_{i=1}^{N}U_{i}(\boldsymbol{a}))$ would exceed the range of the largest predefined real number on a personal computer. However, if we measure the system throughput in the Gbps unit, the parameter $\gamma$ is large and becomes $0.2*1024$, $0.5*1024,$ and $0.85*1024$, respectively.}. We show the dynamics of the time average throughputs of all the APs in Figure \ref{fig:Dynamics-of-distributed1} when $\gamma=0.85$. It demonstrates the convergence of the cooperative channel selection algorithm. From Figure \ref{fig:Dynamics-of-potential1}, we see that the performance of the algorithm improves as the $\gamma$ increases, and the convergence time also increases accordingly. When $\gamma=0.85$, the performance loss of the cooperative channel selection algorithm is less than $1\%$, compared with the centralized optimal solution, i.e., $\max_{\boldsymbol{a}\in\Theta}\sum_{n=1}^{N}U_{n}(\boldsymbol{a})$. Moreover, the algorithm achieves more than $18\%$ performance gain over the random channel selection scheme wherein the APs choose channels purely randomly.

\begin{figure}[tt]
\centering
\includegraphics[scale=0.75]{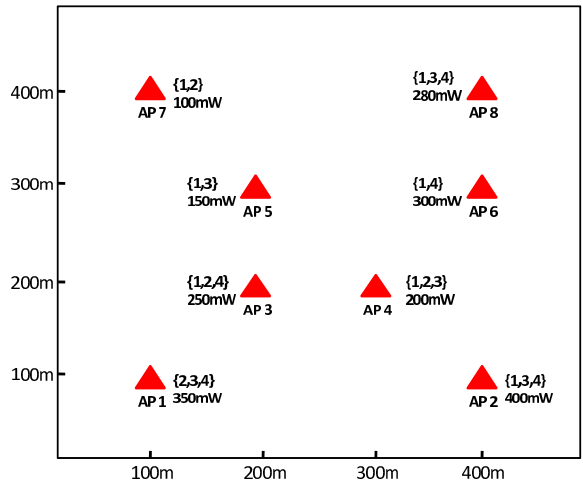}
\caption{\label{fig:A-square-area}A square area of a length of $500$m with
$8$ scattered APs. Each AP has a set of vacant channels, and operates
with a specific transmission power. For example, the available channels
and transmission power of AP $1$ are $\{2,3,4\}$ and $350$ mW,
respectively. }
\end{figure}

\begin{figure}[tt]
\centering
\includegraphics[scale=0.5]{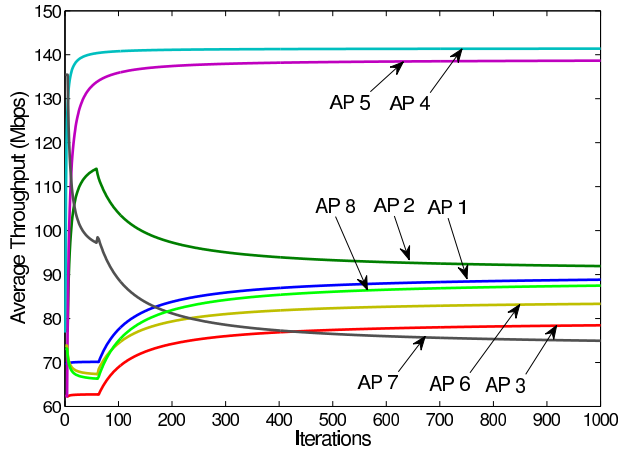}
\caption{\label{fig:Dynamics-of-distributed1}Dynamics of APs'  time average throughputs in cooperative channel selection with $\gamma=0.85$}
\end{figure}

\begin{figure}[tt]
\centering
\includegraphics[scale=0.5]{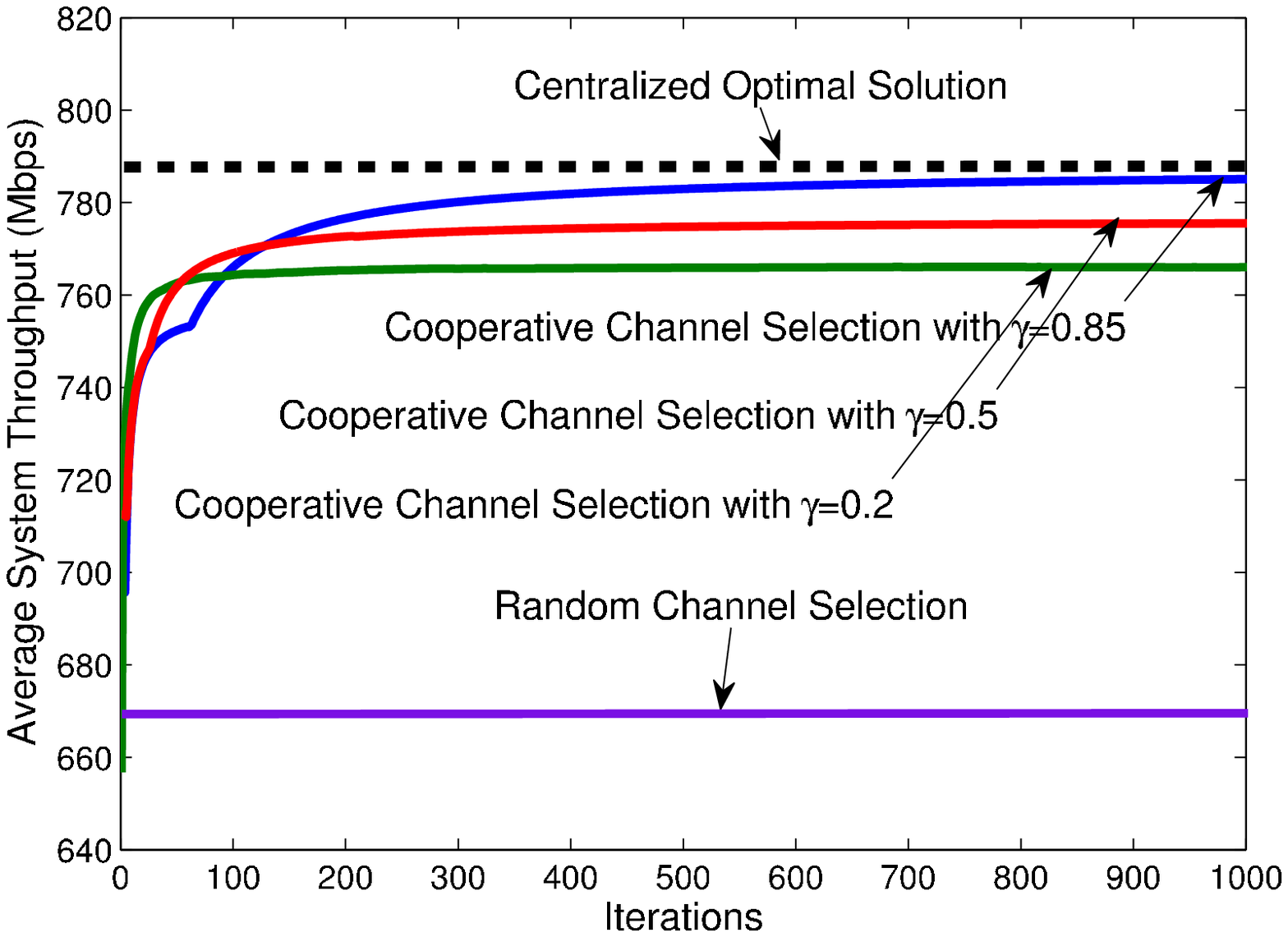}
\caption{\label{fig:Dynamics-of-potential1}Dynamics of time average system throughput}
\end{figure}

\subsection{\label{sub:Distributed-AP-Channel}Non-Cooperative AP Channel Selection}

We then implement the non-cooperative channel selection algorithm in Section \ref{sec:Distributed-AP-Channel}. We show the dynamics of the throughputs of all the APs
in Figure \ref{fig:Dynamics-of-distributed}. We see that the algorithm
converges to an equilibrium $\boldsymbol{a}^{*}$ in less $20$  iterations.
To verify that the equilibrium is a Nash equilibrium, we show the dynamics of the potential function $\Phi$ in Figure \ref{fig:Dynamics-of-potential}.
We see that the algorithm can lead the potential function to a maximum
point, which is a Nash equilibrium according to the property of potential
game. At the equilibrium $\boldsymbol{a}^{*}$, $8$ APs achieve the throughputs
$U_{n}(\boldsymbol{a}^{*})$ of $\{101.4, 100.1, 82.6, 97.6, 83.2, 98.7, 85.6, 84.5\}$
Mbps, respectively, and no AP has the incentive to deviate its channel
selection unilaterally. Compared with cooperative AP channel selection algorithm, the performance loss of the non-cooperative channel selection algorithm is less than $7\%$. Such a performance loss is due to the selfishness of APs in the non-cooperative environment. However, the convergence time of non-cooperative AP channel selection algorithm is much shorter. This is because that in order to achieve the system optimal solution, the cooperative algorithm needs more time to randomly explore the whole set of feasible channel selections.  While the non-cooperative channel selection algorithm achieves the Nash equilibrium by focusing on the subset of channel selections satisfying the finite improvement property.

We then further implement simulations with $N=10,20,...,50$ APs being randomly scattered over the square area in Figure \ref{fig:A-square-area}, respectively. The number of TV channels $M=50$ and $25$ channels out of these $50$ channels will be randomly chosen as the set of vacant channels $\mathcal{M}_{n}$ for each AP $n$. We implement both non-cooperative and cooperative AP channel selection algorithms. The results are shown in Figure \ref{fig:Comparison}. We see that when the number of APs is small (e.g., $N\leq 20$), the non-cooperative channel selection achieves the same performance as the cooperative case. This is due to the abundance of the spectrum resources. We also observe that the performance of the non-cooperative channel selection algorithm is less  than $8\%$ in all cases. This demonstrates the efficiency of the non-cooperative channel selection.

\begin{figure}[tt]
\centering
\includegraphics[scale=0.5]{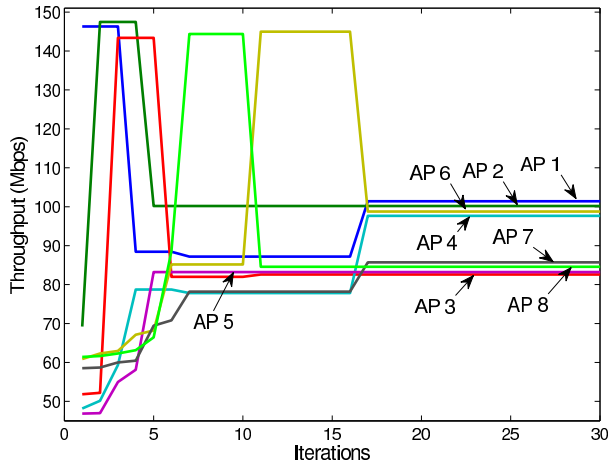}
\caption{\label{fig:Dynamics-of-distributed}Dynamics of non-cooperative AP channel
selection}
\end{figure}

\begin{figure}[tt]
\centering
\includegraphics[scale=0.5]{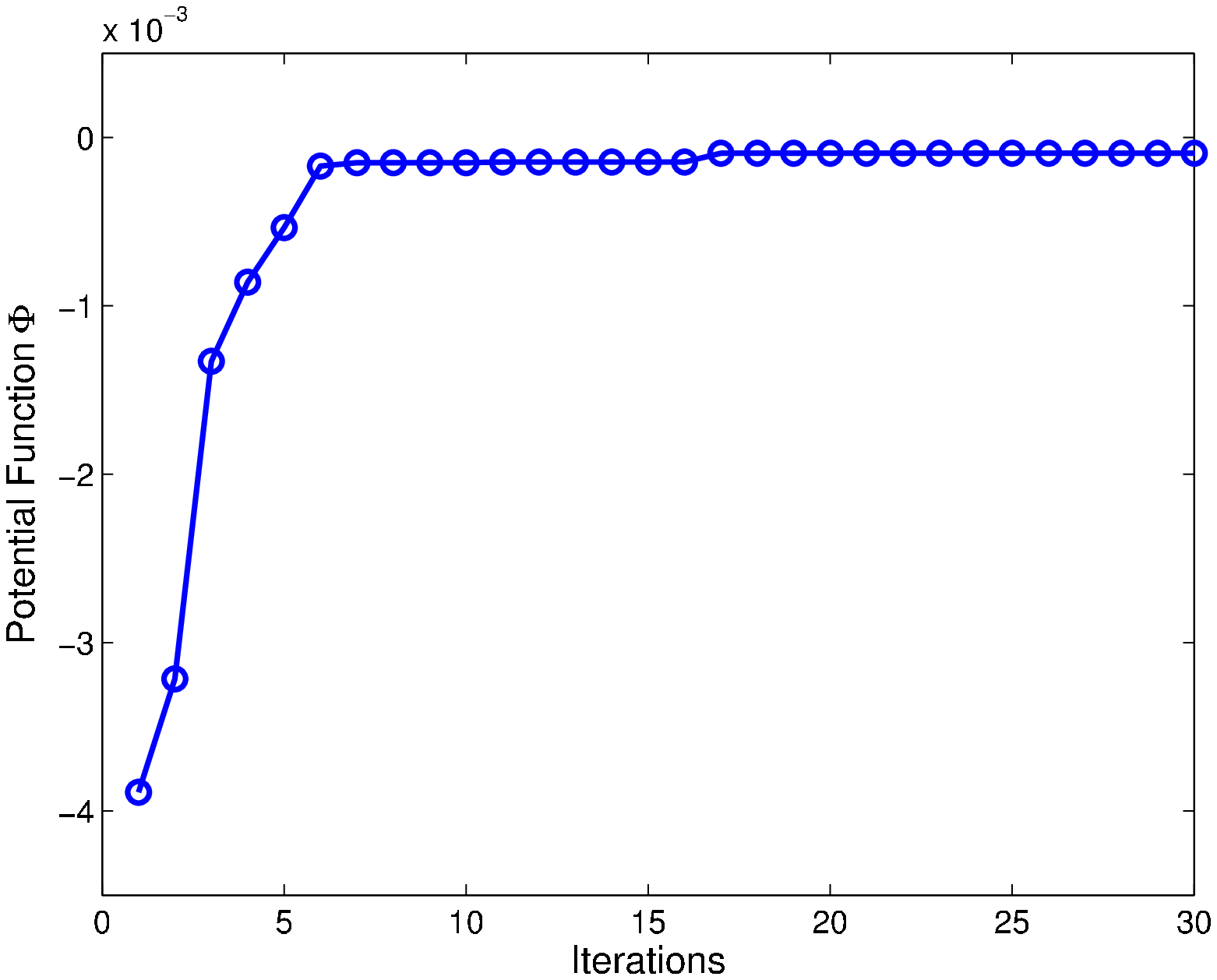}
\caption{\label{fig:Dynamics-of-potential}Dynamics of potential function value
$\Phi$ corresponding to the dynamics in Figure \ref{fig:Dynamics-of-distributed}}
\end{figure}

\begin{figure}[tt]
\centering
\includegraphics[scale=0.5]{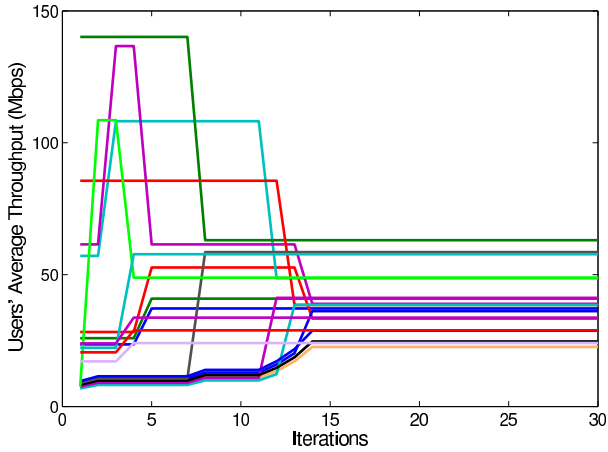}
\caption{\label{fig:Dynamics-of-distributed-1}Dynamics of distributed AP association}
\end{figure}

\begin{figure}[tt]
\centering
\includegraphics[scale=0.5]{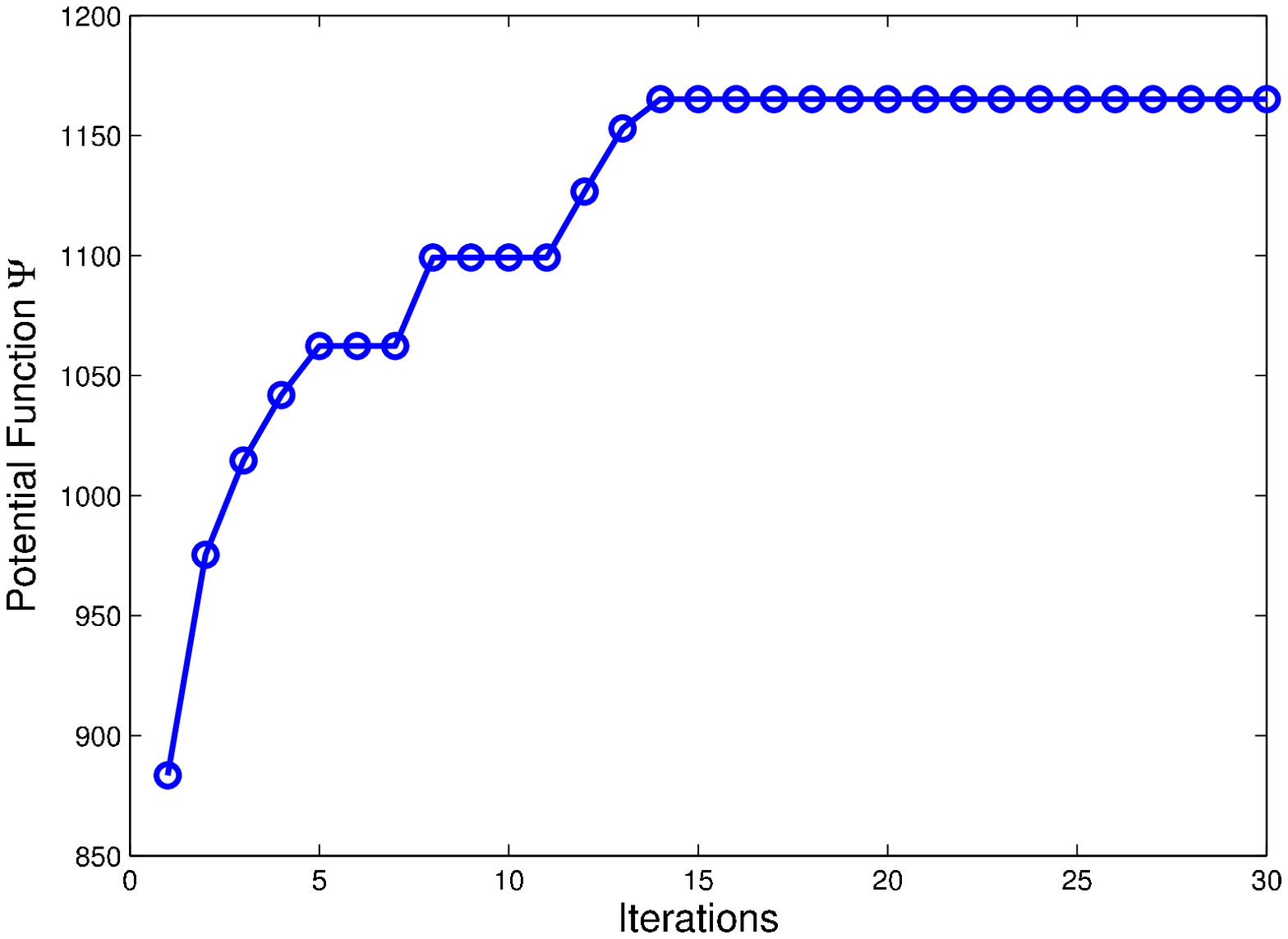}
\caption{\label{fig:Dynamics-of-potential2}Dynamics of potential function
value $\Psi$  corresponding to the dynamics in Figure \ref{fig:Dynamics-of-distributed-1}}
\end{figure}

\begin{figure}[tt]
\centering
\includegraphics[scale=0.6]{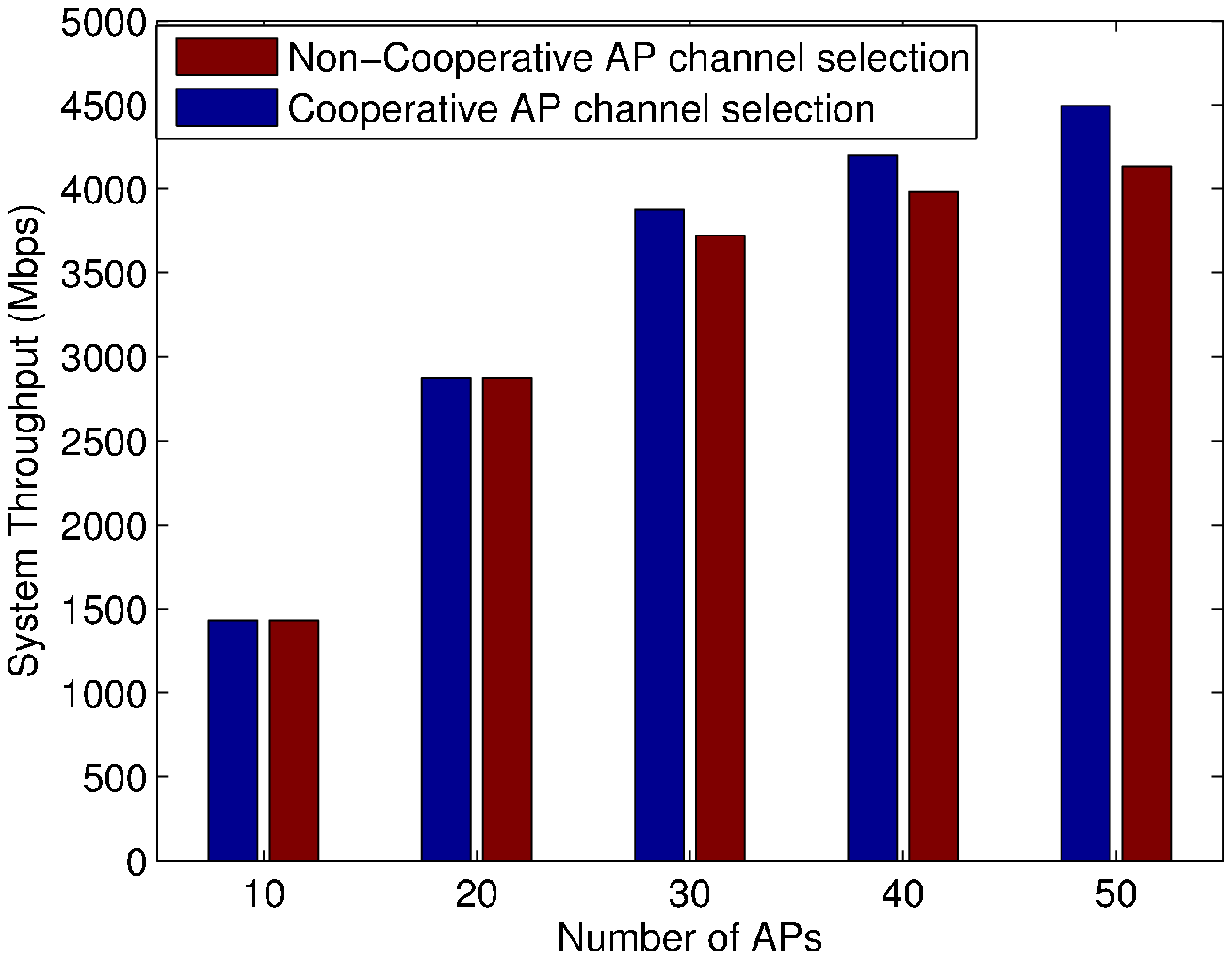}
\caption{\label{fig:Comparison}System throughput by cooperative and non-cooperative AP channel selection with the number of APs $N=10,20,...,50$, respectively.}
\end{figure}

\subsection{Distributed AP Association} We next implement the distributed AP association algorithm in Section
\ref{sec:Distributed-AP-Association-1}. We consider $K=20$ mobile
secondary users who can move around and try to find a proper AP
to associate with. Within an AP $n$, the worse-case throughput $U_{n}(\boldsymbol{a}^{*})$
of AP $n$ is computed according to the
Nash equilibrium $\boldsymbol{a}^{*}$ in Section \ref{sub:Distributed-AP-Channel}.
For the channel contention by multiple secondary users, we set the
number of backoff mini-slots $\lambda_{\max}=10$.

We first show in Figure
\ref{fig:Dynamics-of-distributed-1} the dynamics of the distributed AP association algorithm
with the random initial APs selections, users' transmission gains $H_{n}^{k}$ being randomly selected from the set $\{1.0,1.1,1.2,1.3,1.4,1.5\}$, and the mobility cost factor $\delta_{k}=0.06$ Mbps/m. We see that the algorithm converges
to an equilibrium $(\boldsymbol{b}^{*},\boldsymbol{b}^{*})$ in less $30$  iterations.
We also show the the dynamics of the state-based potential function
$\Psi$ in Figure \ref{fig:Dynamics-of-potential2}. We see that the
equilibrium $(\boldsymbol{b}^{*},\boldsymbol{b}^{*})$ is a state-based Nash equilibrium, since
the algorithm leads the state-based potential function to a maximum
point.

\begin{figure*}[tt]
\centering
\includegraphics[scale=0.75]{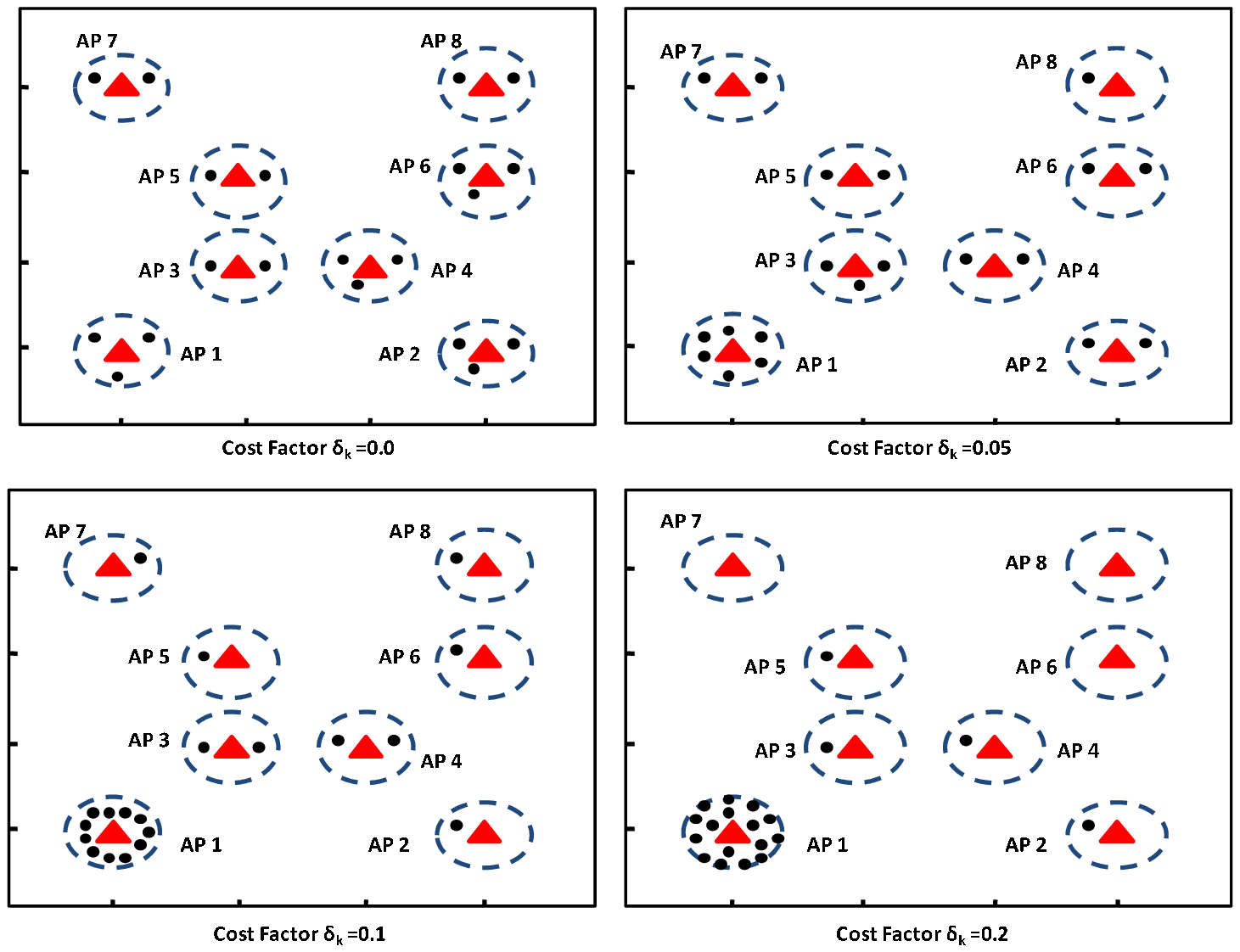}
\caption{\label{fig:The-distribution-of}The equilibrium secondary user distribution
with different cost factors. The black dots represent the secondary
users, and the dash cycles represent the transmission ranges of APs.}
\end{figure*}

To investigate the impact of the cost factor $\delta_{k}$, we assume that
all the users are initially associated with AP $1$ with the same transmission gains $H_{n}^{k}=1$ and they change AP associations according to
the distributed AP association algorithm with four different settings
in Figure \ref{fig:The-distribution-of}.  In each setting, all users have the same cost factor $\delta_{k}$. As the mobility
cost increases, we see that less secondary users are willing to move away from their initial APs.
When the cost $\delta_{k}=0$, the secondary users are scattered across
all APs since there is no cost due to mobility. When $\delta_{k}=0.2$,
only a small fraction of secondary users move away from the initial AP $1$ to the  APs closeby, due to the high cost of mobility. In Figure \ref{fig711}, we further implement the algorithm with a mixture of two types of secondary users: high and low mobility cost factors.  We see that users of low mobility cost will spread out to achieve better data rates, while most users of high mobility cost choose to stay in AP $1$ and suffer from severe congestion.

We next investigate the robustness of the distributed AP association
algorithm. We consider $K=30$ mobile secondary users with the cost
factor $\delta_{k}$ randomly generated from a uniform distribution
in $(0,0.2)$. At iteration $t=200$ and $400$, we let
$10$ users leave the system and $15$ new users enter the system,
respectively. The results in Figures \ref{fig:Dynamics-of-distributed2}
and \ref{fig:Dynamics-of-potential3} show that the algorithm can
quickly converge to a state-based Nash equilibrium after the perturbations
occur. This verifies that the distributed AP association algorithm
is robust to the dynamics of secondary users' leaving and entering the system.

\begin{figure}[tt]
\centering
\includegraphics[scale=0.65]{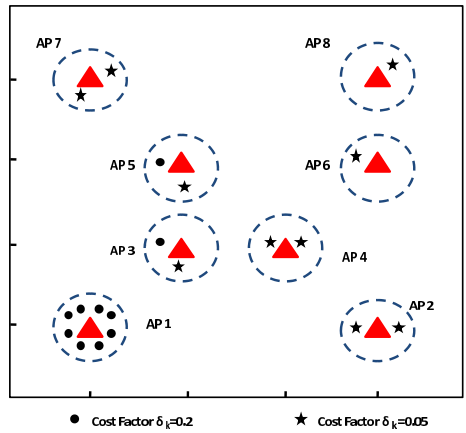}
\caption{\label{fig711}The equilibrium secondary user distribution
with a mixture of two types of secondary users. The black dots represent the secondary
users with a high cost factor, and the black stars represent the secondary
users with a low cost factor.}
\end{figure}

\begin{figure}[tt]
\centering
\includegraphics[scale=0.5]{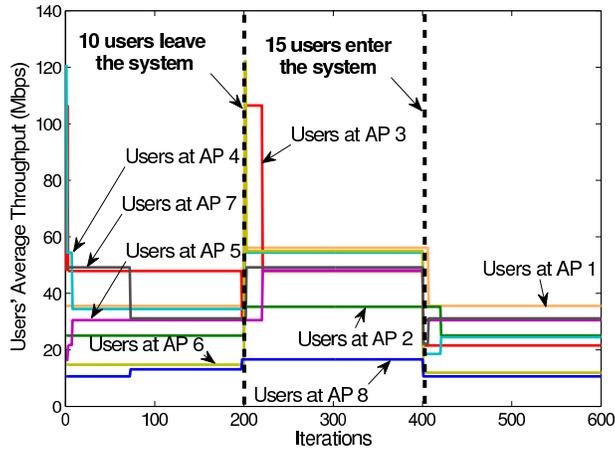}
\caption{\label{fig:Dynamics-of-distributed2}Dynamics of distributed AP association
with perturbations. At iteration $t=200$ and $400$, $10$ users
leave the system and $15$ new users enter the system, respectively.}
\end{figure}

\begin{figure}[tt]
\centering
\includegraphics[scale=0.45]{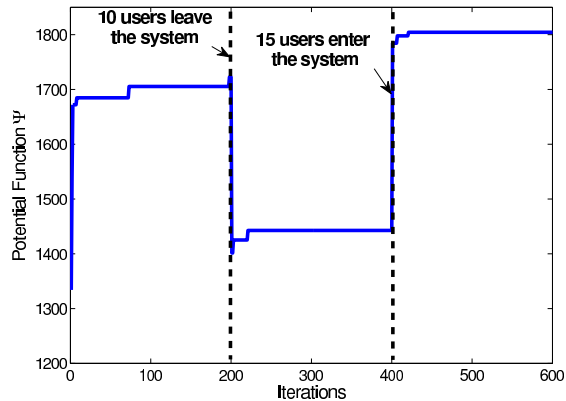}
\caption{\label{fig:Dynamics-of-potential3}Dynamics of potential function
value $\Psi$ with perturbations. At iteration $t=200$ and $400$,
$10$ users leave the system and $15$ new users enter the system,
respectively.}
\end{figure}

%% file: Conclusion.tex
\section{\label{sec:Related-Work}Related Work}

Most research efforts in database-assisted white-space systems are
devoted to the design of geo-location service. Gurney \emph{et al.} in \cite{db1}
calculated the spectrum availability based on the transmission power
of the white-space devices. Karimi in \cite{db2} presented a method
to derive location-specific maximum permitted emission levels for
white space devices. Murty \emph{et al.} in \cite{db4} proposed a framework
to determine the vacant spectrum by using propagation model and
terrain data. Nekovee in \cite{db3} studied the white-space availability
and frequency composition in UK.

For the white-space networking system design, many existing works
focus on the experimental testbed implementation. Bahl \emph{et al.} in \cite{db5}
designed a single white-space AP system. Murty \emph{et al.} in \cite{db4}
addressed the client bootstrapping and mobility handling issues in
white-space AP networks. Feng \emph{et al.} in \cite{db6} considered the
OFDM-based AP white-space network system design. Deb \emph{et al.} in \cite{db7}
presented a centralized white-space spectrum allocation algorithm.
In this paper, we propose a theoretic framework based on game theory
for distributed resource allocation in white-space AP networks.

The game theory has been used to study wireless resource allocation
problems in non-white-space infrastructure-based networks.  Song \emph{et al.} in \cite{ap1} modeled the distributed channel allocation in mesh
networks as a non-cooperative game, where each cell tries to minimize
the interference received from other cells. Southwell \emph{et al.} in \cite{Southwell2012Spectrum} modeled the distributed channel selection problem with switching cost as a network congestion game. Chen and Huang in \cite{Chen2012Spatial} proposed a spatial spectrum access game framework for distributed spectrum sharing with spatial reuse. Wang \emph{et al.} in \cite{wang2010spectrum} proposed an auction approach for incentive-compatible spectrum resource allocation. Most previous works studied the competitive channel selection based on the protocol interference model where two users can interfere with each other if they are linked by an interference edge on the interference graph. In this paper, we explore the competitive channel selections based on the physical interference model, which is not well studied in the literature. The most relevant work is \cite{kauffmann2007measurement}, where Kauffmann \emph{et al.} considered to minimize the total interferences received by all the APs by designating each AP a specific  utility function to be optimized locally. In our paper, we consider the case that each AP is fully rational and tries to maximize its own throughput.

For the AP association problem, Gaji\'{c} \emph{et al.} in \cite{ap2}
and Duan et. al in \cite{ap4} studied the pricing mechanisms to achieve
efficient wireless service provider association solutions. Bejerano \emph{et al.} in \cite{bejerano2004fairness} address the load imbalance problem through the association control.  Hong \emph{et al.} in \cite{ap3} investigated distributed AP association game with power control, by assuming that the chosen channels among APs are
non-overlapping. These previous results focus on the case that users are stationary, and can associate with any AP. When users are mobile, Mittal \emph{et al.} in \cite{ap5} studied the distributed access
point selection game by assuming that users are homogeneous
with the same cost of mobility. Here we propose a state-based game
framework to formulate the more general case that users have heterogeneous
cost of mobility.

\section{\label{sec:Conclusion}Conclusion}

In this paper, we consider the database-assisted white-space AP network
design. We address the cooperative and non-cooperative channel selection problems among the APs and
the distributed AP association problem of the secondary users. We propose the cooperative and non-cooperative AP channel selection
algorithms and a distributed AP association algorithm, all of which that converge
to the corresponding equilibrium globally. Numerical results show that the proposed algorithms are efficient, and are also robust to
the perturbation by secondary users' dynamical leaving and entering
the system.

For the future work, we are going to generalize the results to the mixture
case that consists of both cooperative and non-cooperative APs. Multiple APs that belong to one network operator are cooperative with each other, but they may not cooperate with other APs that belong to a different network operator. It will be interesting to study the existence of Nash equilibrium and design distributed algorithms to achieve the equilibrium.

Although the distributed AP association algorithm can achieve the state-based Nash equilibrium wherein all users are satisfied given their mobility cost factors, the loads among different APs can be quite imbalanced when the mobility cost is high as demonstrated in the numerical results. Thus, how to design an incentive compatible mechanism such as pricing to achieve load balance among the APs with mobile secondary
users will be very interesting and challenging.